\documentclass[11pt]{article}

\usepackage[margin=1in]{geometry}
\usepackage[T1]{fontenc}
\usepackage[utf8]{inputenc}
\usepackage{lmodern}
\usepackage{microtype}
\usepackage{amsmath,amssymb,mathtools}
\usepackage{xcolor}
\usepackage{booktabs}
\usepackage{tabularx}
\usepackage{array}
\usepackage{tikz}
\usepackage{longtable}
\usepackage{caption}

\usepackage{array}
\usepackage{booktabs}
\usepackage{float}
\usepackage{subcaption}
\usepackage{tabularx}
\definecolor{linkblue}{RGB}{30,70,140}
\usepackage[
  colorlinks=true,
  linkcolor=linkblue,
  citecolor=linkblue,
  urlcolor=linkblue,
  pdfauthor={Ravi Tandon},
  pdftitle={Strengthening Recursive Constructions for \\Zero-Error Shannon Capacity}
]{hyperref}

\usepackage[toc,page,titletoc]{appendix}
\usepackage{titlesec}

\usepackage{amsthm}

\theoremstyle{plain}
\newtheorem{theorem}{Theorem}
\newtheorem{lemma}{Lemma}
\newtheorem{proposition}{Proposition}

\theoremstyle{definition}
\newtheorem{definition}{Definition}

\theoremstyle{definition}
\newtheorem{remark}{Remark}

\newcommand{\strong}{\boxtimes}
\newcommand{\Shannon}{\Theta}

\usepackage{booktabs}
\usepackage{tabularx}
\usepackage{array}
\usepackage{xcolor}
\usepackage{sourcesanspro}

\definecolor{preAI}{HTML}{7A3E48}   
\definecolor{AIera}{HTML}{24557A}   

\definecolor{preAI}{HTML}{4A1721}   
\definecolor{AIera}{HTML}{102E4A}   

\newcommand{\preai}[1]{\textcolor{preAI}{#1}}
\newcommand{\aiera}[1]{\textcolor{AIera}{#1}}

\makeatletter
\def\bstctlcite{\@ifnextchar[{\@bstctlcite}{\@bstctlcite[@auxout]}}
\def\@bstctlcite[#1]#2{%
  \@bsphack
  \@for\@citeb:=#2\do{%
    \edef\@citeb{\expandafter\@firstofone\@citeb}%
    \if@filesw
      \immediate\write\csname #1\endcsname{%
        \string\citation{\@citeb}}%
    \fi
  }%
  \@esphack
}
\makeatother

\title{Strengthening Recursive Constructions for \\Zero-Error Shannon Capacity}

\author{
Ravi Tandon\\
School of Electrical, Computing and Software Engineering,\\
University of Arizona, Tucson, AZ 85721, USA.\\
Email: \href{mailto:tandonr@arizona.edu}{\texttt{tandonr@arizona.edu}}}
\date{}

\begin{document}
\bstctlcite{BSTcontrol}

\maketitle
\vspace{-2.5em}
\begin{abstract}

The exact Shannon capacity remains unknown for every odd cycle beyond the
five-cycle $C_5$, making odd cycles central examples in zero-error
information theory. Improving the known lower bounds requires constructing
large independent sets in strong powers of these graphs. Recent AI-assisted
work has accelerated this effort and produced a rapid sequence of
improvements. In particular, building on the construction of Itty et al.,
Gao developed a recursive product construction for combining structured
independent sets, and Buys, Polak, and Zuiddam (BPZ) subsequently strengthened
this through a richer 
recursion framework.

We continue this line of AI-assisted exploration and introduce a
\emph{heterogeneous refinement} of these recursive constructions. The
central observation is that the usefulness of an intermediate construction
depends not only on the size of its current main independent set, but also
on the auxiliary structure that it carries into subsequent recursion.
Consequently, different parts of the auxiliary structure need not always
use the same independent set, and different occurrences in a recursive
construction need not always use the same intermediate representation. We
first formalize this principle for Gao's binary product and derive explicit
propagation rules showing how heterogeneous choices can strengthen the
resulting gadget while leaving its current code size unchanged. We then
extend this principle to the more general BPZ recursive framework, showing
how different heterogeneous constructions can be tailored to the distinct
roles they play within the recursion.

Applying these refinements to the seven-cycle $C_7$, we obtain an independent set in $C_7^{\strong 500}$ yielding
$\Shannon(C_7)\ge 3.25883262\ldots$, thereby improving the best known lower bound. Beyond the numerical improvement, the results illustrate a general principle for recursive zero-error code constructions: intermediate structures having the same dimension and current code size can have different downstream value depending on where and how they are used in the recursion.
\end{abstract}

\vspace{0pt}
\section{Introduction}\label{sec:introduction}
\vspace{0pt}

In classical information theory, reliable communication usually means that
the probability of decoding error can be made arbitrarily small by using long
codes. Shannon's theory of \emph{zero-error} communication asks for something
strictly stronger: which communication rates are possible when confusion is
not allowed at all?~\cite{Shannon1956} The problem is naturally represented by
a graph $G$, called the confusability graph. Each possible channel input is represented by a vertex, and two
vertices are connected when the corresponding inputs can produce a common
output and hence may be confused by the receiver. A collection of inputs that
are pairwise nonconfusable is therefore an independent set in this graph.

When the channel is used $d$ times, a channel input is a word
$\mathbf{x}=(x_1,\ldots,x_d)$. 
Two words are confusable if, in every coordinate, the corresponding symbols are either equal or adjacent in the confusability graph $G$.
Graph-theoretically, the relevant graph is the
$d$-fold strong power $G^{\strong d}$. If $\alpha(H)$ denotes the independence
number (size of the largest independent set) of a graph $H$, then the largest zero-error code of blocklength $d$ has
size $\alpha(G^{\strong d})$. The Shannon capacity of $G$ is
$\Shannon(G)
  =\sup_{d\ge 1}\alpha\!\left(G^{\strong d}\right)^{1/d}$.
Thus, an independent set of size $M$ in $G^{\strong d}$ immediately gives the
lower bound $\Shannon(G)\ge M^{1/d}$. The normalization by the $d$th root is
important: it converts a blocklength-$d$ code into the effective number of
zero-error symbols transmitted per channel use.

\vspace{-7pt}
\paragraph{The first unresolved odd cycle: \(C_7\).}
Cycle graphs are among the oldest and most natural examples in zero-error
information theory. The classical nontrivial example is the pentagon $C_5$.
Shannon exhibited five pairwise nonconfusable two-letter words, proving
$\Shannon(C_5)\ge\sqrt{5}$~\cite{Shannon1956}. Lov\'asz later introduced
the theta function $\vartheta(G)$, proved the general upper bound
$\Shannon(G)\le\vartheta(G)$, and established
$\Shannon(C_5)=\vartheta(C_5)=\sqrt{5}$~\cite{Lovasz1979}. 
The situation changes immediately at the next odd cycle. For every odd cycle
$C_{2m+1}$ with $m\ge 3$, the exact Shannon capacity remains unknown. In
particular, $C_7$ is the smallest cycle whose Shannon capacity has not been
determined. Lov\'asz's theta function gives the explicit upper bound
  $\Shannon(C_7)
  \le \vartheta(C_7)
  =\frac{7\cos(\pi/7)}{1+\cos(\pi/7)}
  =3.31766720\ldots$,
which remains the best known upper bound for \(C_7\) to date. 
The challenge is therefore to construct increasingly large independent sets in strong powers of \(C_7\), thereby improving the best known lower bounds and narrowing the gap to the Lov\'asz upper bound.

\vspace{-7pt}
\paragraph{Overview of progress on lower bounds for \(C_7\).}
The history of lower bounds for $C_7$ (and, more generally, for larger odd cycles) illustrates both the difficulty of the problem and the importance of finding structure at the right blocklength. It can be divided into two periods. From 1971 through 2019, progress came through increasingly strong finite-dimensional constructions~\cite{BaumertEtAl1971,VeselZerovnik2002,MathewOstergard2017}, culminating in the $367$-word independent set in $C_7^{\strong 5}$ of Polak and Schrijver~\cite{PolakSchrijver2019}. Beginning in July 2026, AI-assisted exploration and structured product constructions produced a rapid sequence of improvements: Itty et al.~\cite{IttyEtAl2026} improved on the direct product in dimension $10$, Gao~\cite{Gao2026} converted that improvement into a binary recursive mechanism, and BPZ~\cite{BuysPolakZuiddam2026,BuysPolakZuiddamGitHub2026} subsequently introduced a stronger base structure and a richer recursion. Table~\ref{tab:C7progress} summarizes this progression. These developments are reviewed in detail in Sections~\ref{subsec:pre-llm}, \ref{subsec:llm-assistedprogress}, and \ref{subsec:BPZmultigadget}.

A striking feature of the recent improvements is the role of large language models (LLMs) in mathematical exploration. Earlier LLM-based program search had already rediscovered the $367$-word lower bound and found new constructions for related odd cycles~\cite{RomeraParedesEtAl2024}. Itty et al.~\cite{IttyEtAl2026} report that their $134753$-word construction was discovered through iterative interactions with an LLM and was not reached by the conventional search heuristics they tested. Gao~\cite{Gao2026} likewise reports using an LLM to implement the computational search and expand the proofs underlying his recursive construction. BPZ~\cite{BuysPolakZuiddam2026} state that, building on these approaches and using ChatGPT 5.6 Sol Pro and Claude Opus 5, they obtained their further improved bounds.

\vspace{10pt}
\noindent\textbf{Zero-error Shannon capacity as a benchmark for AI-assisted mathematical discovery.} Taken together, these developments suggest that zero-error Shannon capacity provides a particularly clean benchmark for AI-assisted mathematical discovery, with odd cycles offering a compelling test case.
The objective is unambiguous, progress is measured by an exact and independently verifiable quantity, and every claimed improvement must be supported by an explicit construction. At the same time, the value of the benchmark is not limited to the final numerical gain: even a small improvement may require new structural ideas, more effective search strategies, and new methods of proof and verification. The process may therefore reveal substantially more about the capabilities of AI-assisted mathematical reasoning than the magnitude of improvement alone.

\begin{table}[t]
\centering
\captionsetup{
    labelfont=bf,
    textfont=bf,
    justification=centering,
    singlelinecheck=false
}
\caption{Progression of lower bounds on the zero-error Shannon capacity of \(C_7\).\\[2pt]
{\normalfont (All displayed decimal values are truncated.)}}
\label{tab:C7progress}
\vspace{-8pt}
{\rmfamily
\small
\renewcommand{\arraystretch}{1.24}
\setlength{\tabcolsep}{5.2pt}

\begin{tabularx}{\linewidth}{
@{}
>{\centering\arraybackslash}p{0.070\linewidth}
>{\raggedright\arraybackslash}p{0.275\linewidth}
>{\raggedright\arraybackslash}X
>{\raggedleft\arraybackslash}p{0.145\linewidth}
@{}
}

\toprule
\textbf{Year}
&
\textbf{Authors}
&
\textbf{Construction / main idea}
&
\(\boldsymbol{\Theta(C_7)\geq}\)
\\
\midrule

\multicolumn{4}{@{}l}{
    \textit{\preai{\textbf{Finite-dimensional constructions}}}
}
\\[2pt]

1971 &
Baumert, McEliece, Rodemich, Rumsey, Stanley, and Taylor
\cite{BaumertEtAl1971} &
\preai{\(343\) words in \(C_7^{\strong 5}\)} &
\preai{\(3.21409584\ldots\)}
\\

2002 &
Vesel and \v{Z}erovnik
\cite{VeselZerovnik2002} &
\preai{\(108\) words in \(C_7^{\strong 4}\) via simulated annealing} &
\preai{\(3.22370979\ldots\)}
\\

2017 &
Mathew and \"Osterg\r{a}rd
\cite{MathewOstergard2017} &
\preai{\(350\) words in \(C_7^{\strong 5}\) using symmetry-constrained stochastic search} &
\preai{\(3.22710880\ldots\)}
\\

2019 &
Polak and Schrijver
\cite{PolakSchrijver2019} &
\preai{\(367\) words in \(C_7^{\strong 5}\) using a circular-graph construction and computation} &
\preai{\(3.25786596\ldots\)}
\\

\addlinespace[5pt]

\multicolumn{4}{@{}l}{
    \textit{\aiera{\textbf{AI-assisted exploration and recursive constructions}}}
}
\\[2pt]

2026 &
Itty, Rosin, Carstensen, and Reichman
\cite{IttyEtAl2026} &
\aiera{\(134753\) words in \(C_7^{\strong 10}\); improvement beyond the direct product} &
\aiera{\(3.25802073\ldots\)}
\\

2026 &
Gao
\cite{Gao2026} &
\aiera{Binary recursive propagation of structured independent sets} &
\aiera{\(3.25878915\ldots\)}
\\

2026 &
Buys, Polak, and Zuiddam (BPZ)
\cite{BuysPolakZuiddam2026} &
\aiera{Stronger base structure followed by recursive propagation} &
\aiera{\(3.25880536\ldots\)}
\\

2026 &
BPZ, subsequent update
\cite{BuysPolakZuiddamGitHub2026} &
\aiera{More general combining rules, including combinations of more than two inputs} &
\aiera{\(3.25882798\ldots\)}
\\

\midrule

2026 &
\textbf{This work (Theorem~\ref{thm:mainC7})} &
\textbf{\aiera{Heterogeneous recursive constructions}} &
\aiera{\(\boldsymbol{3.25883262\ldots}\)}
\\

\bottomrule
\end{tabularx}
}
\vspace{-10pt}
\end{table}

\vspace{8pt}
\noindent\textbf{Our contribution and organization of the paper.}
This paper further develops the recursive-construction viewpoint underlying the recent advances in lower bounds on zero-error Shannon capacity~\cite{IttyEtAl2026,Gao2026,BuysPolakZuiddam2026,BuysPolakZuiddamGitHub2026}. Our emphasis is both structural and expository. Rather than presenting our numerical improvement in isolation, we use the recent sequence of constructions to identify and formalize additional flexibility in the recursive framework, and to explain how it can be exploited. The central principle is that the value of an intermediate construction is determined not only by the size of its main independent set at that stage, but also by the auxiliary structure that it carries into later stages of the recursion.

We first give a self-contained account of Gao's binary product construction and the gadget structure that it propagates. We then introduce a \textit{heterogeneous refinement} (Theorem~\ref{thm:heteroGao}) in which different classes of the right-hand auxiliary set may be paired with different left-hand independent sets, derive the corresponding propagation rules, and show that this already improves the bounds obtained from Gao's recursion~\cite{Gao2026} and the first BPZ refinement~\cite{BuysPolakZuiddam2026}. We next describe the BPZ multi-gadget framework~\cite{BuysPolakZuiddamGitHub2026}, including admissible combining rules, and show how Gao's binary product arises as a special case. Finally, we construct role-specific heterogeneous Gao gadgets, convert them into BPZ seven-family representations, and use them in different branches of the BPZ recursion for $C_7$. This yields an independent set in
$C_7^{\strong 500}$ and, as presented in Theorem~\ref{thm:mainC7}, the improved
lower bound
\begin{align}\label{eq:main-bound-intro}
\Shannon(C_7)
\ge
3.25883262\ldots .
\end{align}
Throughout, we provide explicit constructions and propagation rules, together with finite certificates and exact arithmetic, so that the constructions and resulting bounds can be reproduced and independently verified.

\vspace{-3pt}
\section{Problem Statement and Background}
\label{sec:background}

\vspace{-2pt}
\subsection{Zero-error communication and Shannon capacity}

The notion of zero-error communication was introduced by Shannon in his
seminal 1956 paper~\cite{Shannon1956}. Consider a communication channel with
a finite input alphabet. Two input symbols are said to be \emph{confusable}
if they can produce a common output at the receiver. These relations can be
represented by a graph \(G\), called the \emph{confusability graph}: the
vertices of \(G\) are the input symbols, and two distinct vertices are joined
by an edge whenever the corresponding symbols are confusable.

Suppose first that the channel is used only once. A collection of symbols can
be transmitted with zero error precisely when no two symbols in the
collection are confusable. In graph-theoretic language, such a collection is
an independent set in \(G\). The largest number of messages that can be
transmitted with zero error in one channel use is therefore the independence
number \(\alpha(G)\).

The problem becomes more interesting when the channel is used repeatedly.
After \(d\) uses, a transmitted message is a word
\(\mathbf{x}=(x_1,\ldots,x_d)\), where each coordinate is an input symbol. Two such
words are confusable if, in every coordinate, the corresponding symbols are
either equal or confusable. The resulting
confusability graph is the \(d\)-fold strong power \(G^{\boxtimes d}\).
Consequently, the largest zero-error code of blocklength \(d\) has size
\(\alpha(G^{\boxtimes d})\). To compare codes of different blocklengths, their sizes must be normalized.
The Shannon capacity of \(G\) is defined as
\begin{align}
    \Theta(G)
    &=
    \sup_{d\geq 1}
    \alpha\!\left(G^{\boxtimes d}\right)^{1/d}.
    \label{eq:shannon-capacity}
\end{align}
Thus, an independent set of size \(M\) in \(G^{\boxtimes d}\) gives the
lower bound \(\Theta(G)\geq M^{1/d}\). The subtlety is that repeated channel uses can perform better than simply
taking independent copies of a one-use code. The pentagon provides the
simplest and most celebrated example of this phenomenon.

\begin{figure}[t]
    \centering
    \includegraphics[width=0.95\linewidth]{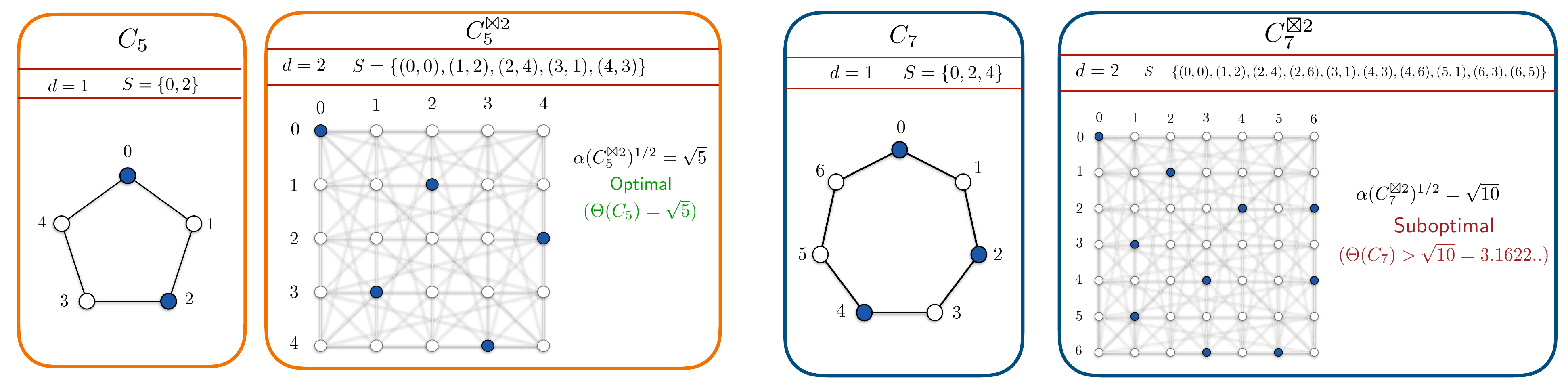}
    \caption{One- and two-use zero-error codes for $C_5$ and $C_7$. The five-word code in $C_5^{\boxtimes 2}$ achieves $\Theta(C_5)=\sqrt{5}$, whereas the ten-word code in $C_7^{\boxtimes 2}$ gives the lower bound $\Theta(C_7)\ge \sqrt{10}$. The exact Shannon capacity of $C_7$ remains unknown.}
    \label{fig:c5-zero-error}
\end{figure}

\noindent \textit{The pentagon $C_5$: a completely solved example}. 
Let $C_5$ denote the cycle on the five symbols
$\{0,1,2,3,4\}$, where each symbol is confusable with its two neighbors.
With one channel use, the largest zero-error code has size
$\alpha(C_5)=2$. Repeating such a code independently over two uses would
therefore produce only $2^2=4$ codewords. As illustrated in
Fig.~\ref{fig:c5-zero-error}, however, joint coding permits five pairwise
nonconfusable codewords in $C_5^{\strong 2}$, yielding
$\Shannon(C_5)\ge\sqrt{5}$. Lovász proved the matching upper bound
$\Shannon(C_5)\le\vartheta(C_5)=\sqrt{5}$, and hence
$\Shannon(C_5)=\sqrt{5}$~\cite{Lovasz1979}. Thus, the pentagon is the
simplest example in which joint coding improves upon independent repetition,
and the gain already achieved at blocklength two is exactly optimal.


\subsection{The seven-cycle $C_7$: the first open odd cycle graph}

We now consider the seven-cycle \(C_7\). Its vertices are identified with
\(\mathbb Z_7=\{0,1,\ldots,6\}\), with arithmetic understood modulo \(7\).
Two distinct symbols are adjacent, and hence confusable, when they differ by
\(1\) or \(-1\) modulo \(7\). For convenience, we include equality in the confusability relation. For
\(a,b\in\mathbb Z_7\), write \(a\simeq b\) when
\(a-b\in\{-1,0,1\}\pmod 7\). For two words
\(x,y\in\mathbb Z_7^d\), we similarly write \(x\simeq y\) when
\(x_i\simeq y_i\) in every coordinate. Thus, two distinct words are
confusable precisely when \(x\simeq y\).

With one use of the channel, three symbols can be selected without
confusion. For example, \(\{0,2,4\}\) is an independent set. Four symbols
cannot be selected without including two adjacent vertices, so
\(\alpha(C_7)=3\). Using this three-symbol code independently over two channel uses would give
\(3^2=9\) codewords. As with the pentagon, joint coding performs better:
there exists a ten-word independent set in \(C_7^{\boxtimes 2}\) (Fig.~\ref{fig:c5-zero-error}). Moreover,
ten is optimal (by optimal we mean that the size of the largest independent set in $C_7^{\boxtimes 2}$ is $10$). Hence
    $\alpha(C_7)= 3, ~~ 
    \alpha\!\left(C_7^{\boxtimes 2}\right) = 10$.
It follows that
\(\Theta(C_7)\geq\sqrt{10}=3.1622776601\ldots\), which exceeds the
one-use value \(3\). The two-use behavior of \(C_5\) and \(C_7\) therefore looks superficially
similar. In both cases, joint coding beats the direct product of an optimal
one-use code. The essential difference is that the two-use construction
settles the pentagon completely, whereas for \(C_7\) still better normalized
code sizes arise at larger blocklengths.

Lov\'asz's theta function gives the general upper bound
\(\Theta(G)\leq\vartheta(G)\). For the seven-cycle, this becomes
\begin{align}
    \Theta(C_7)
    &\leq
    \vartheta(C_7)
    =
    \frac{7\cos(\pi/7)}
         {1+\cos(\pi/7)}
    =
    3.317667207394095\ldots .
    \label{eq:c7-theta}
\end{align}
As the higher-dimensional constructions reviewed below show,
the Shannon capacity of $C_7$ is strictly larger than $\sqrt{10}$. Determining how far the lower
bound can be pushed requires the construction of increasingly large
independent sets in higher strong powers of \(C_7\). The next section reviews the sequence of constructions that progressively
improved this lower bound.

\subsection{Progress before LLM-assisted search (1971-2019)}
\label{subsec:pre-llm}

The size-$10$ independent set in $C_7^{\boxtimes 2}$ described above gives $\Theta(C_7)\geq 10^{1/2}\approx 3.16228$. A natural way to improve this bound is to search for larger independent sets in higher strong powers of $C_7$. If $I\subseteq V(C_7^{\boxtimes d})$ is an independent set, then
$\Theta(C_7)\geq |I|^{1/d}$.
Thus, when comparing constructions across different dimensions, the relevant quantity is the normalized size $|I|^{1/d}$ rather than simply $|I|$. 

\vspace{10pt}
\noindent \textbf{[1971: $\Theta(C_7)\geq 343^{1/5}\approx 3.21410$]}: In 1971, Baumert, McEliece, Rodemich, Rumsey, Stanley, and Taylor~\cite{BaumertEtAl1971} found the independence number for $C_7^{\boxtimes 3}$
and showed that 
$\alpha(C_7^{\boxtimes 3})=33$. Their approach viewed a vertex $(x_1,x_2,x_3)\in\mathbb{Z}_7^3$ as the
corner of a $2\times2\times2$ cube in the discrete $7^3$-torus, so that
an independent set in $C_7^{\boxtimes 3}$ becomes a packing of disjoint
cubes. By slicing the torus into seven two-dimensional layers and using
$\alpha(C_7^{\boxtimes 2})=10$, they first obtained the upper bound
$\alpha(C_7^{\boxtimes 3})\leq35$. They then used a symmetry-reduced
exhaustive computer search over compatible layer packings to rule out
packings of sizes $35$ and $34$, while finding several packings of size
$33$. Thus, their construction led to the following lower bound on $\Theta(C_7)$: $\Theta(C_7)\geq 33^{1/3}\approx3.20753$. The same paper went on to give a stronger lower bound on the Shannon capacity
through an explicit construction in the fifth strong power $(d=5)$. Baumert et
al.~\cite{BaumertEtAl1971} showed that the $343=7^3$ vectors
\[
\left(
x_1,x_2,x_3,\,
2x_1+2x_2+2x_3,\,
2x_1+4x_2+6x_3
\right),
\qquad x_1,x_2,x_3\in\mathbb{Z}_7,
\]
form an independent set in $C_7^{\boxtimes 5}$. To see this, consider
the difference $(a,b,c)$ between the first three coordinates of two
codewords. If the two words were confusable, then
$a,b,c\in\{0,\pm1\}$ and also
$2(a+b+c)\in\{0,\pm1\}$ modulo $7$. This forces
$a+b+c\in\{0,\pm3\}$. The cases $a+b+c=\pm3$ are separated by the
fifth coordinate, while if $a+b+c=0$, any nonzero possibility is a
permutation of $(1,-1,0)$ and again gives a fifth-coordinate difference
outside $\{0,\pm1\}$. Hence distinct codewords are nonconfusable. This
construction yields $\Theta(C_7)\geq 343^{1/5}\approx 3.21410$, which was stronger than the bound obtained from the exact three-dimensional result.

\vspace{10pt}
\noindent \textbf{[2002: $\Theta(C_7)\geq 108^{1/4}\approx 3.22371$]}: Vesel and \v{Z}erovnik~\cite{VeselZerovnik2002} subsequently used simulated annealing to construct an independent set of size $108$ in $C_7^{\boxtimes 4}$, yielding $\Theta(C_7)\geq 108^{1/4}\approx 3.22371$.

\vspace{10pt}
\noindent \textbf{[2017: $\Theta(C_7)\geq 350^{1/5}\approx 3.22711$]}:
The next improvements came from the fifth strong power. Mathew and \"{O}sterg{\aa}rd~\cite{MathewOstergard2017} used stochastic search with prescribed symmetries to obtain an independent set of size $350$ in $C_7^{\boxtimes 5}$, giving $\Theta(C_7)\geq 350^{1/5}\approx 3.22711$.

\vspace{10pt}
\noindent \textbf{[2019: $\Theta(C_7)\geq 367^{1/5}\approx 3.257866$]}:
A substantially larger improvement was obtained by Polak and Schrijver~\cite{PolakSchrijver2019}. Using an auxiliary circular graph and computer search, they constructed an independent set
$
I\subseteq V(C_7^{\boxtimes 5})$, with $|I|=367$,
and hence established $\Theta(C_7)\geq 367^{1/5}\approx 3.257866$.


These advances increasingly relied on computational search, but the underlying paradigm remained the same: choose a dimension $d$, search for a large independent set in $C_7^{\boxtimes d}$, and use its cardinality to lower-bound $\Theta(C_7)$. The $367$-word construction remained the benchmark for subsequent progress. The developments described next introduced a different possibility: exploiting the \emph{internal structure} of a known independent set to  construct larger independent sets in higher dimensions.

\vspace{-3pt}
\subsection{Overview of Recent LLM-assisted Progress}\label{subsec:llm-assistedprogress}

The recent progress on $C_7$ began with the work of Itty, Rosin, Carstensen, and Reichman~\cite{IttyEtAl2026}. Starting from the $367$-word independent set $I_{0}\subseteq C_7^{\boxtimes 5}$ of Polak and Schrijver~\cite{PolakSchrijver2019}, the direct Cartesian product $I_{0}\times I_{0}$ gives an independent set of size $367^2=134689$ in $C_7^{\boxtimes 10}$. Itty et al. improved upon this by modifying the product construction locally: they remove eight carefully chosen words from $I_{0}$, leaving a set $B$ of size $359$, retain the large core $B\times B$, and replace the remaining part of the Cartesian product by two specially constructed families, each of size $8\cdot367$. The resulting independent set has size $359^2+2\cdot8\cdot367=134753$, and therefore yields $\Theta(C_7)\ge 134753^{1/10}>3.258020$. The authors report that this construction was discovered through iterative interactions with an LLM, which generated and executed search procedures that successively improved the size of the ten-dimensional code. This construction is important not only for the numerical improvement, but also because it revealed that the $367$-word code could be exploited in a more structured manner than by taking a simple Cartesian product, an observation that Gao subsequently developed into a general recursive construction as we describe next.

\vspace{-3pt}
\subsubsection{Gao's Product Construction}\label{subsec:GaoproductandC7}

Gao~\cite{Gao2026} observed that the construction of Itty et al.\ \cite{IttyEtAl2026} contains additional structure that can be preserved under further products. His key contribution was to isolate this structure, encode it through a small collection of parameters, and prove a product lemma showing that two such structured independent sets can be combined to produce another object of the same type. This converts the ten-dimensional improvement of Itty et al.\ from a one-time construction into a recursive mechanism that can be iterated to increasingly large dimensions. We next give a self-contained description of Gao's construction and product lemma. Alongside the general development, we use a simple running example based on $C_7$ to illustrate each of the constituent objects and to show how the product construction recovers an optimal $10$-word independent set in $C_7^{\boxtimes 2}$.

For vertices $u,v$ of a graph $G$, write $u\simeq v$ when $u=v$ or $u$ and $v$ are adjacent, so that $u\simeq v$ means that the two vertices are confusable. For $S\subseteq V(G)$, let
\begin{align}
N(S)=\{u\in V(G):u\simeq v\text{ for some }v\in S\}
\end{align}
denote the closed neighborhood of $S$. For a singleton, we write $N(v)=N(\{v\})$.

\begin{definition}[Private pair]
Let $I\subseteq V(G)$ be an independent set. A pair $(r,q)$ of vertices of $G$ is called a \emph{private pair} for $I$ if
\begin{align}
r\in I,\qquad q\notin I,\qquad N(q)\cap I=\{r\}.
\end{align}
The vertex $r$ is called the \emph{center} of the private pair, and $q$ is called its \emph{private neighbor}.
Equivalently, $q$ is confusable with exactly one vertex of $I$, namely $r$.
\end{definition}

\noindent \textit{Example $1$: To illustrate these definitions throughout this subsection, we use the following simple example. Consider the independent set $I=\{0,2,4\}\subseteq V(C_7)$. Take $r=0$ and $q=6$. The closed neighborhood of $6$ is $N(6)=\{5,6,0\}$, and hence $N(6)\cap I=\{0\}$. Therefore, $(0,6)$ is a private pair for $I$: the vertex $0$ is its center and $6$ is its private neighbor. We will continue to use this example as the remaining components of Gao's gadget and the subsequent product construction are introduced.}

\begin{definition}[Gao gadget and profile]\label{def:GaoGadget}
Let $G$ be a graph. A \emph{Gao gadget} in $G$ consists of an independent set $I\subseteq V(G)$ together with a selected collection of $t$ pairwise endpoint-disjoint private pairs
\begin{align}
\{(r_i,q_i)\}_{i=1}^{t},
\end{align}
two independent sets $P^{\mathrm H}$ and $P^{\mathrm V}$ forming complementary transversals of these private pairs, and an auxiliary independent set $X\subseteq V(G)$ satisfying
\begin{align}
X\cap N(P^{\mathrm H})\cap N(P^{\mathrm V})=\varnothing.
\end{align}
Here, complementary transversals means that, for every $i$, one endpoint of the pair $\{r_i,q_i\}$ belongs to $P^{\mathrm H}$ and the other belongs to $P^{\mathrm V}$.  Thus each private pair contributes exactly one endpoint to each transversal.  The final condition requires that no vertex of $X$ is simultaneously
confusable with a vertex of $P^{\mathrm H}$ and a vertex of
$P^{\mathrm V}$.
Associated with the selected private pairs, define
\begin{align}
R&=\{r_1,\ldots,r_t\}, \quad 
Q=\{q_1,\ldots,q_t\}, \quad 
B=I\setminus R.
\end{align}
Thus, $R$ is the set of centers, $Q$ is the set of private neighbors, and $B$ is the part of $I$ remaining after the selected centers are removed. The auxiliary set $X$ is then partitioned according to its interaction with the two transversals. Define
\begin{align}
X^0
&=X\setminus\left(N(P^{\mathrm H})\cup N(P^{\mathrm V})\right), \nonumber\\
X^{\mathrm H}
&=X\cap N(P^{\mathrm H}), \nonumber\\
X^{\mathrm V}
&=X\cap N(P^{\mathrm V}).
\end{align}
Since no point of $X$ is confusable with both transversals, these three sets are pairwise disjoint and
\begin{align}
X=X^0\mathbin{\dot\cup}X^{\mathrm H}\mathbin{\dot\cup}X^{\mathrm V}.
\end{align}
\begin{figure}[t]
    \centering
\includegraphics[width=0.95\linewidth]{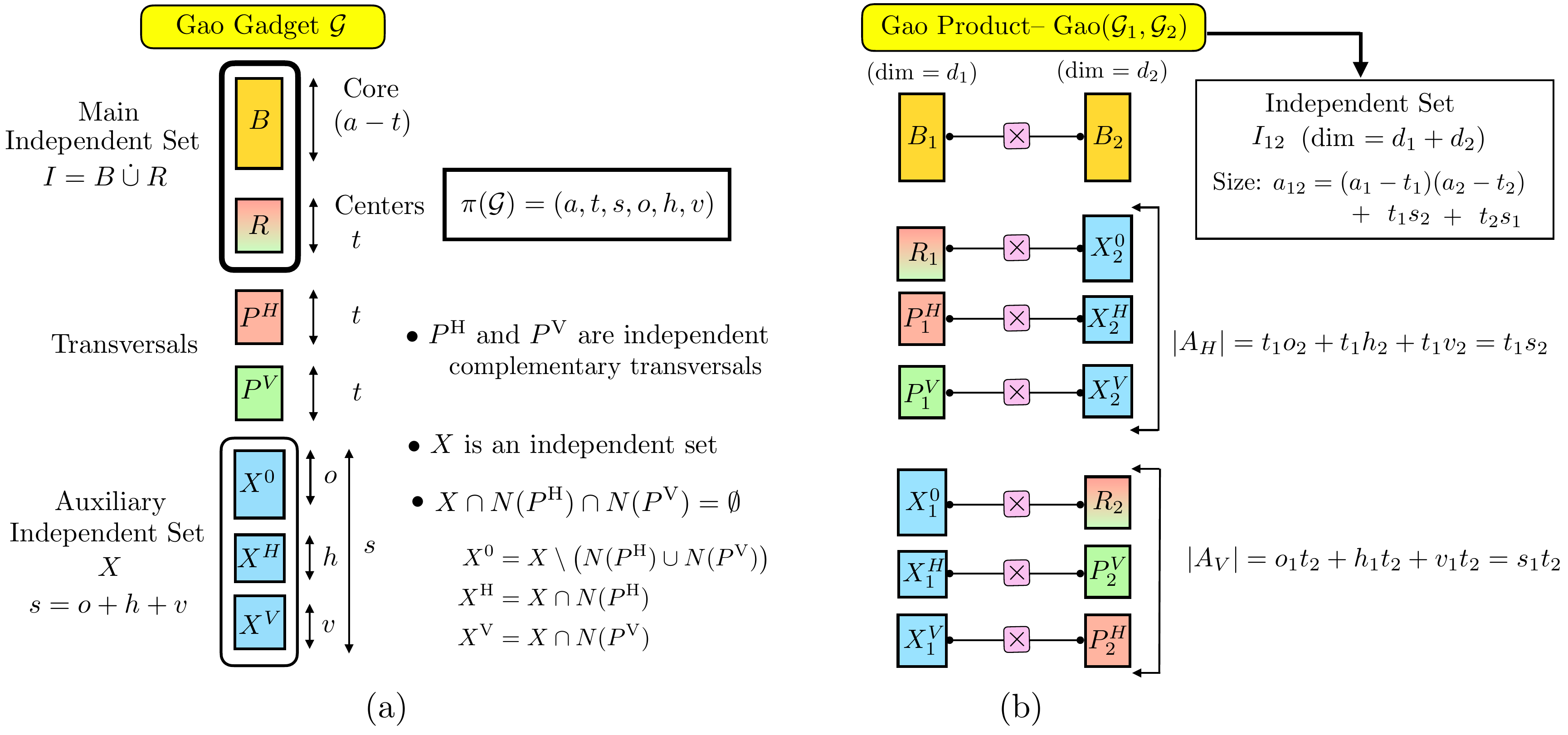}
    \caption{Schematic description of a Gao gadget and Gao's product construction. \textup{(a)} A Gao gadget $\mathcal G$ with six-parameter profile $\pi(\mathcal G)=(a,t,s,o,h,v)$, consisting of a main independent set $I=B\mathbin{\dot\cup}R$, two independent complementary transversals $P^{\mathrm H}$ and $P^{\mathrm V}$, and an auxiliary independent set $X=X^0\mathbin{\dot\cup}X^{\mathrm H}\mathbin{\dot\cup}X^{\mathrm V}$ satisfying $X\cap N(P^{\mathrm H})\cap N(P^{\mathrm V})=\emptyset$. \textup{(b)} Gao's product applied to gadgets $\mathcal G_1$ and $\mathcal G_2$. Each pink $\times$ denotes a Cartesian-product block. The three blocks $R_1\times X_2^0$, $P_1^{\mathrm H}\times X_2^{\mathrm H}$, and $P_1^{\mathrm V}\times X_2^{\mathrm V}$ form $A_{\mathrm H}$, while the three blocks $X_1^0\times R_2$, $X_1^{\mathrm H}\times P_2^{\mathrm V}$, and $X_1^{\mathrm V}\times P_2^{\mathrm H}$ form $A_{\mathrm V}$. Together with $B_1\times B_2$, these blocks form the independent set
$I_{12}=(B_1\times B_2)\mathbin{\dot\cup}A_{\mathrm H}\mathbin{\dot\cup}A_{\mathrm V}$
in dimension $d_1+d_2$, of size
$a_{12}=(a_1-t_1)(a_2-t_2)+t_1s_2+t_2s_1$.}
    \label{fig:gaogadget}
\end{figure} 
Thus, the vertices in $X^0$ are confusable with neither transversal, those in $X^{\mathrm H}$ are confusable with $P^{\mathrm H}$ but not with $P^{\mathrm V}$, and those in $X^{\mathrm V}$ are confusable with $P^{\mathrm V}$ but not with $P^{\mathrm H}$. The \emph{profile} of the Gao gadget is the six-tuple (see Fig.~\ref{fig:gaogadget} for an illustration):
\begin{align}
\pi(\mathcal G) \triangleq (a,t,s,o,h,v),
\end{align}
where
$a=|I|,  
s=|X|,  
o=|X^0|, 
h=|X^{\mathrm H}|, 
v=|X^{\mathrm V}|$,
and $s=o+h+v$. 
\end{definition}

\noindent \textit{Example $1$ (continued): Continuing the example above, the independent set $I=\{0,2,4\}$ together with the private pair $(0,6)$ gives $t=1$. We choose the complementary transversals $P^{\mathrm H}=\{0\}$ and $P^{\mathrm V}=\{6\}$, and take $X=\{1,3,5\}$. Relative to these transversals, $1$ is confusable with $P^{\mathrm H}$ but not with $P^{\mathrm V}$, $5$ is confusable with $P^{\mathrm V}$ but not with $P^{\mathrm H}$, and $3$ is confusable with neither. Hence $R=\{0\}$, $Q=\{6\}$, $B=\{2,4\}$, $X^{\mathrm H}=\{1\}$, $X^0=\{3\}$, and $X^{\mathrm V}=\{5\}$. The resulting Gao gadget therefore has profile $\pi(\mathcal G)=(3,1,3,1,1,1)$.}

\begin{lemma}[Gao's product lemma~\cite{Gao2026}] \label{lemma:GaoProduct}
Let $\mathcal G_1$ and $\mathcal G_2$ be Gao gadgets in graphs $G_1$ and $G_2$, respectively, with profiles
\begin{align}
\pi(\mathcal G_i)=(a_i,t_i,s_i,o_i,h_i,v_i),
\qquad i\in\{1,2\}.
\end{align}
For $x\in X_2$, define
\begin{align}
L_1(x)=
\begin{cases}
R_1, & x\in X_2^0,\\
P_1^{\mathrm H}, & x\in X_2^{\mathrm H},\\
P_1^{\mathrm V}, & x\in X_2^{\mathrm V},
\end{cases}
\end{align}
and for $y\in X_1$, define
\begin{align}
K_2(y)=
\begin{cases}
R_2, & y\in X_1^0,\\
P_2^{\mathrm V}, & y\in X_1^{\mathrm H},\\
P_2^{\mathrm H}, & y\in X_1^{\mathrm V}.
\end{cases}
\end{align}
Define the two routed families
\begin{align}
A_{\mathrm H}
&=\{(p,x):x\in X_2,\ p\in L_1(x)\}, \nonumber\\
A_{\mathrm V}
&=\{(y,q):y\in X_1,\ q\in K_2(y)\},
\end{align}
and let
\begin{align}
I_{12}
=
(B_1\times B_2)
\cup A_{\mathrm H}
\cup A_{\mathrm V}.
\end{align}
Then $I_{12}$ is an independent set in $G_1\boxtimes G_2$.

Moreover, $I_{12}$ carries the following selected collection of pairwise endpoint-disjoint private pairs:
\begin{align}
&\{((r_i,x),(q_i,x)):
1\leq i\leq t_1,\ x\in X_2^0\}
 \cup
\{((y,r_j),(y,q_j)):
y\in X_1^0,\ 1\leq j\leq t_2\}.
\end{align}
Complementary transversals for these private pairs are
\begin{align}
P_{12}^{\mathrm H}
&=
(P_1^{\mathrm H}\times X_2^0)
\cup
(X_1^0\times P_2^{\mathrm V}),
\nonumber\\
P_{12}^{\mathrm V}
&=
(P_1^{\mathrm V}\times X_2^0)
\cup
(X_1^0\times P_2^{\mathrm H}).
\end{align}
Taking
\begin{align}
X_{12}=X_1\times X_2
\end{align}
as the auxiliary independent set, its decomposition relative to the two propagated transversals is
\begin{align}
X_{12}^0
&=
(X_1^0\times X_2^0)
\mathbin{\dot\cup}
\left(
(X_1^{\mathrm H}\cup X_1^{\mathrm V})
\times
(X_2^{\mathrm H}\cup X_2^{\mathrm V})
\right),
\nonumber\\
X_{12}^{\mathrm H}
&=
(X_1^{\mathrm H}\times X_2^0)
\mathbin{\dot\cup}
(X_1^0\times X_2^{\mathrm V}),
\nonumber\\
X_{12}^{\mathrm V}
&=
(X_1^{\mathrm V}\times X_2^0)
\mathbin{\dot\cup}
(X_1^0\times X_2^{\mathrm H}).
\end{align}
Consequently, these objects form a new Gao gadget $\mathcal G_{12}$ in $G_1\boxtimes G_2$, whose profile is
\begin{align}
\pi(\mathcal G_{12})
=
(a_{12},t_{12},s_{12},o_{12},h_{12},v_{12}),
\end{align}
where
\begin{align}
a_{12}
&=(a_1-t_1)(a_2-t_2)+t_1s_2+s_1t_2,
\nonumber\\
t_{12}
&=t_1o_2+o_1t_2,
\nonumber\\
s_{12}
&=s_1s_2,
\nonumber\\
o_{12}
&=o_1o_2+(h_1+v_1)(h_2+v_2),
\nonumber\\
h_{12}
&=h_1o_2+o_1v_2,
\nonumber\\
v_{12}
&=v_1o_2+o_1h_2.
\end{align}
\end{lemma}

\begin{figure}[t]
    \centering
\includegraphics[width=0.95\linewidth]{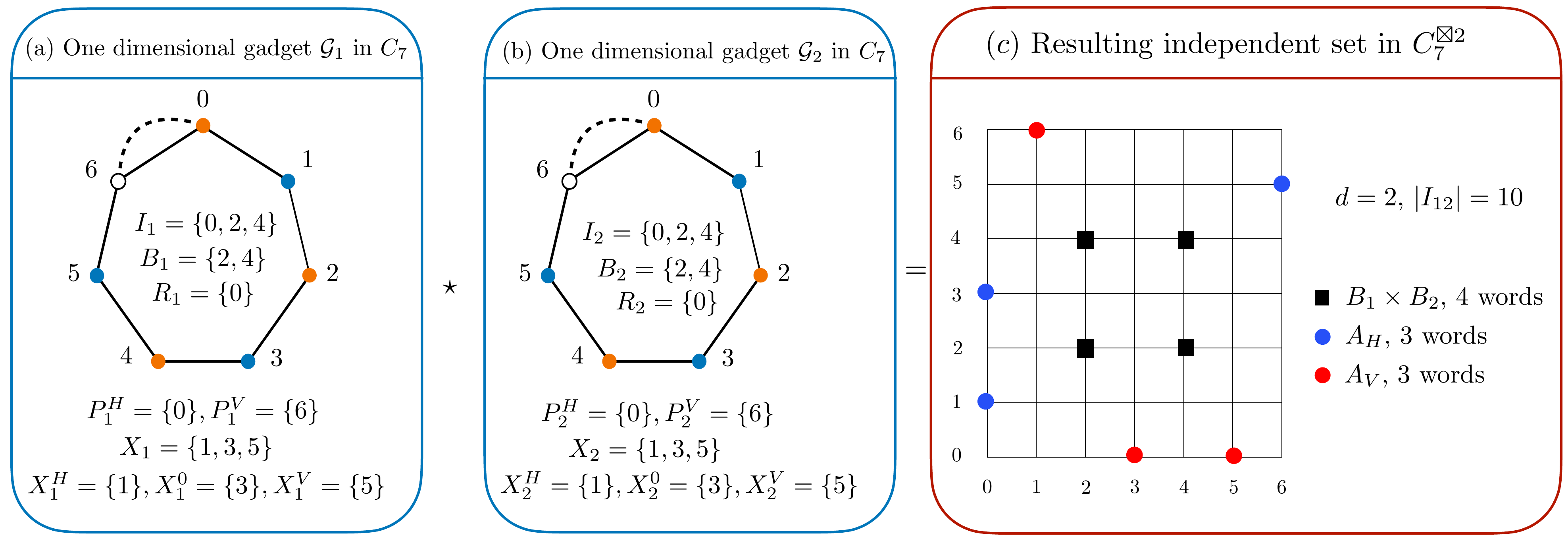}
    \caption{Illustration of Gao's product construction on two one-dimensional gadgets in $C_7$. The Cartesian-product core $B_1\times B_2$ is augmented by the routed families $A_H$ and $A_V$, producing an optimal $10$-word independent set in $C_7^{\boxtimes 2}$.}
    \label{fig:c7-gao-example}
\end{figure} 

\noindent\textit{Example $1$ (continued).}
\textit{We now apply Gao's product lemma to two copies of the gadget in our example, each with profile $\pi(\mathcal G)=(3,1,3,1,1,1)$. Recall that $B=\{2,4\}$, $R=\{0\}$, $P^{\mathrm H}=\{0\}$, $P^{\mathrm V}=\{6\}$, and $X^{\mathrm H}=\{1\}$, $X^0=\{3\}$, $X^{\mathrm V}=\{5\}$. The first routing map therefore gives $A_{\mathrm H}=\{(0,1),(0,3),(6,5)\}$, while the reversed routing in the second coordinate gives $A_{\mathrm V}=\{(1,6),(3,0),(5,0)\}$. Together with the Cartesian-product core $B_1\times B_2=\{(2,2),(2,4),(4,2),(4,4)\}$, this produces the $10$-word independent set $I_{12}=\{(2,2),(2,4),(4,2),(4,4),(0,1),(0,3),(6,5),(1,6),(3,0),(5,0)\}$ in $C_7^{\boxtimes 2}$. Thus, Gao's product construction recovers an optimal $10$-word code.
More importantly, the resulting code again carries the structure of a Gao gadget. Since $X_1^0=X_2^0=\{3\}$, the construction propagates the two private pairs $((0,3),(6,3))$ and $((3,0),(3,6))$. Taking $X_{12}=X_1\times X_2$ as the new auxiliary set, its three classes have sizes $|X_{12}^0|=5$, $|X_{12}^{\mathrm H}|=2$, and $|X_{12}^{\mathrm V}|=2$. Hence the resulting gadget has profile $\pi(\mathcal G_{12})=(10,2,9,5,2,2)$. This illustrates the central feature of Gao's lemma: the product does not merely produce a larger independent set, but produces another gadget of the same type, so that the construction can be applied recursively.}

\vspace{5pt}

\noindent
\textbf{Gao's Construction for Improving $C_7$}. For $C_7$, Gao constructs the five-dimensional base gadget from the $367$-word independent set $I_0\subseteq C_7^{\boxtimes5}$ of Polak and Schrijver~\cite{PolakSchrijver2019}. Gao uses the same eight pairs $(r_j,q_j)$ that appear in the construction of Itty et al.~\cite{IttyEtAl2026}, and verifies that each is a private pair for $I_0$. 
The two complementary transversals are
\begin{equation}
\begin{aligned}
P^{\mathrm H}
&= \{r_0,r_5,r_6\}\cup\{q_1,q_2,q_3,q_4,q_7\},\\
P^{\mathrm V}
&= \{q_0,q_5,q_6\}\cup\{r_1,r_2,r_3,r_4,r_7\}.
\end{aligned}
\end{equation}
For $w=(w_0,\ldots,w_4)\in\mathbb Z_7^5$, Gao defines the automorphism $T(w)=(2-w_1,w_3,w_0,2-w_2,w_4)$, takes $T(I_0)$, and replaces the single word $(2,4,6,3,5)$ by $(1,5,6,3,5)$ to obtain the auxiliary set $X$. A finite computer verification establishes that $I_0$ and $X$ are independent sets of size $367$, that the eight pairs are private, that the two transversals are independent, and that the points of $X$ split into $321$ points confusable with neither transversal, $26$ confusable only with $P^{\mathrm H}$, and $20$ confusable only with $P^{\mathrm V}$. Thus Gao's five-dimensional base gadget has profile
\begin{align}
\text{Profile of gadget~} \mathcal G_5:  \pi(\mathcal{G}_5)=(367,8,367,321,26,20).
\end{align}
Applying the product lemma to two copies of this gadget gives $(367-8)^2+2\cdot8\cdot367=134753$, reproducing the ten-dimensional construction of Itty et al. \cite{IttyEtAl2026}, but now with the additional gadget structure needed for further recursion.

Gao then iterates the product, using the output gadget from one stage as an input to subsequent stages. Let $\mathcal G_5$ denote the five-dimensional base gadget described above. Since the product of gadgets in dimensions $d_1$ and $d_2$ produces a new gadget in dimension $d_1+d_2$, Gao recursively constructs
\begin{align}
\mathcal G_{10} &= \operatorname{Gao}(\mathcal G_5,\mathcal G_5) \rightarrow 
\mathcal G_{15} = \operatorname{Gao}(\mathcal G_5,\mathcal G_{10}) \rightarrow 
\mathcal G_{25} = \operatorname{Gao}(\mathcal G_{10},\mathcal G_{15}), \nonumber\\
\mathcal G_{50} &= \operatorname{Gao}(\mathcal G_{25},\mathcal G_{25}) \rightarrow 
\mathcal G_{100} = \operatorname{Gao}(\mathcal G_{50},\mathcal G_{50}) \rightarrow 
\mathcal G_{200} = \operatorname{Gao}(\mathcal G_{100},\mathcal G_{100}),
\end{align}
where $\operatorname{Gao}(\mathcal{G}_1, \mathcal{G}_2)$ denotes the product operation of Gao's Lemma \ref{lemma:GaoProduct}. Thus the final $200$-dimensional gadget is obtained from forty copies of the original five-dimensional gadget. Gao selects this product tree by optimizing over the possible binary splits at each stage using exact integer computations. The resulting code in $C_7^{\boxtimes 200}$ yields
\begin{align}
\Theta(C_7)\geq 3.258789153908691016\ldots .
\end{align}
The important point for our purposes is the recursive nature of the construction: the product lemma does not merely produce a larger independent set, but produces another gadget of the same form, which can itself be combined again.

\paragraph{Stronger base gadget and Lean verification by BPZ.}

Buys, Polak, and Zuiddam~\cite{BuysPolakZuiddam2026} subsequently found a stronger five-dimensional base gadget. Specifically, Gao's base gadget has profile $(367,8,367,321, 26, 20)$, whereas their new gadget has profile $(367,8,367,322, 26, 19)$. Thus, the code size, number of private pairs, and auxiliary-set size remain unchanged, while the number of auxiliary words confusable with neither transversal increases from $321$ to $322$. They then apply the same recursive product construction, combining gadgets according to the sequence $1+1=2$, $1+2=3$, $2+3=5$, followed by repeated squaring $5+5=10$, $10+10=20$, and $20+20=40$. Starting from five-dimensional gadgets, this produces a gadget in $C_7^{\boxtimes 200}$ and improves Gao's bound to $\Theta(C_7)\geq 3.258805369885\ldots$. A further contribution of Buys, Polak, and Zuiddam is that the resulting bounds, including the validity of the base gadgets and the recursive product calculations, are fully formalized in Lean; Gao's original construction was instead verified by an accompanying Python program using exact finite checks and integer arithmetic.

\section{Heterogeneous Refinement of Recursive Gadgets}\label{sec:heterogeneousGao}

\paragraph{Motivation for Improving Gadget Structure.}
Gao's recursion highlights an important feature of structured product constructions: the usefulness of a gadget is not determined solely by the size of its current code. The other parameters of the gadget describe additional structure---private pairs, the auxiliary independent set, and its interaction with the two transversals---that is propagated into subsequent product steps. These parameters directly influence the size and structure of later descendants. Consequently, two gadgets with the same code size may behave very differently under further recursion, and a gadget that offers no immediate improvement in code size may nevertheless become more valuable at a later stage if it carries more favorable auxiliary structure.

This suggests that, when constructing the output of Gao's product lemma, one should ask not only how large the new independent set is, but also whether the accompanying gadget structure can be chosen more effectively for subsequent products. Our main observation is that Gao's construction imposes a certain uniformity on the auxiliary set of the product gadget that is not always necessary. By allowing different parts of this auxiliary structure to be built from different independent sets whenever the required separation conditions permit, we obtain a more flexible way of completing the same product code into a gadget. We refer to this additional freedom as \emph{heterogeneity}. Importantly, heterogeneity may leave the current code size unchanged while producing more favorable auxiliary structure, whose advantage becomes visible only in later recursive steps.

\paragraph{Overview of Heterogeneous Gadget Construction.}
Consider two Gao gadgets $\mathcal G_L$ and $\mathcal G_R$ in graphs $G_L$ and $G_R$, respectively, where the subscripts $L$ and $R$ denote the \textit{left} and \textit{right} gadgets. Gao's product lemma first constructs the product code $I_{LR}$, together with its propagated private pairs and complementary transversals. The auxiliary set is then chosen to be the Cartesian product $X_L\times X_R$. Using the decomposition $X_R=X_R^0\mathbin{\dot\cup}X_R^{\mathrm H}\mathbin{\dot\cup}X_R^{\mathrm V}$, Gao's choice can equivalently be written as
\begin{align}
X_L\times X_R
=
(X_L\times X_R^0)
\mathbin{\dot\cup}
(X_L\times X_R^{\mathrm H})
\mathbin{\dot\cup}
(X_L\times X_R^{\mathrm V}).
\end{align}
Thus, the same left-hand auxiliary set $X_L$ is used over all three classes of the right-hand auxiliary set. The first key observation is that these three classes play \textit{different roles} with respect to the propagated transversals. Recall that Gao's product construction gives
\begin{align}
P_{LR}^{\mathrm H}
&=
(P_L^{\mathrm H}\times X_R^0)
\cup
(X_L^0\times P_R^{\mathrm V}), \nonumber\\
P_{LR}^{\mathrm V}
&=
(P_L^{\mathrm V}\times X_R^0)
\cup
(X_L^0\times P_R^{\mathrm H}).
\end{align}
Consider first a point whose right coordinate lies in $X_R^{\mathrm H}$. Such a point is confusable with $P_R^{\mathrm H}$ but not with $P_R^{\mathrm V}$. Moreover, since $X_R$ is independent, a point of $X_R^{\mathrm H}$ is nonconfusable with every point of $X_R^0$. Consequently, its right coordinate already separates it from both product components involving $X_R^0$, as well as from the component $X_L^0\times P_R^{\mathrm V}$ of $P_{LR}^{\mathrm H}$. The only possible interaction that remains is with $X_L^0\times P_R^{\mathrm H}$, which belongs to $P_{LR}^{\mathrm V}$. Thus, regardless of the choice of its left coordinate, such a product point can be confusable with at most one of the two propagated transversals. The situation for a right coordinate in $X_R^{\mathrm V}$ is symmetric. Such a point is nonconfusable with $P_R^{\mathrm H}$ and with every point of $X_R^0$, so the only possible interaction is with $X_L^0\times P_R^{\mathrm V}$, which belongs to $P_{LR}^{\mathrm H}$. Again, the right coordinate itself guarantees that the product point cannot be confusable with both propagated transversals.  The class $X_R^0$ imposes a different requirement. If $x\in X_R^0$, then $x$ is nonconfusable with both $P_R^{\mathrm H}$ and $P_R^{\mathrm V}$, so the two components involving $P_R^{\mathrm H}$ and $P_R^{\mathrm V}$ create no restriction. A product point $(u,x)$ can therefore be confusable with both propagated transversals only if $u$ is confusable with both $P_L^{\mathrm H}$ and $P_L^{\mathrm V}$. Hence the left-hand codebook over $X_R^0$ need not equal $X_L$; it need only be an independent set satisfying the same auxiliary-set separation condition as $X_L$.

We may therefore choose three independent sets $J_0,J_{\mathrm H},J_{\mathrm V}\subseteq V(G_L)$, where $J_0$ additionally satisfies
\begin{align}
J_0\cap N(P_L^{\mathrm H})\cap N(P_L^{\mathrm V})
=
\varnothing,
\end{align}
and define the heterogeneous auxiliary set
\begin{align}
X_{LR}^{\mathrm{het}}
=
(J_0\times X_R^0)
\mathbin{\dot\cup}
(J_{\mathrm H}\times X_R^{\mathrm H})
\mathbin{\dot\cup}
(J_{\mathrm V}\times X_R^{\mathrm V}).
\end{align}
Thus, the auxiliary set of the output gadget is no longer forced to inherit the same left-hand codebook over all three parts of $X_R$. The three pieces above remain mutually separated because their right-hand coordinates belong to distinct subsets of the independent set $X_R$, while independence within each piece follows from the independence of $J_0$, $J_{\mathrm H}$, and $J_{\mathrm V}$. Hence this additional freedom affects only the auxiliary structure of the output gadget: Gao's product code $I_{LR}$, its propagated private pairs, and its transversals are left unchanged.

The potential gain now comes from three possibly distinct choices. The set $J_0$ may have a more favorable size and interaction with $P_L^{\mathrm H}$ and $P_L^{\mathrm V}$ than the original auxiliary set $X_L$, while $J_{\mathrm H}$ and $J_{\mathrm V}$ may have more favorable size and interaction with $X_L^0$. These quantities determine how the new auxiliary set is distributed among its three classes and therefore determine the profile of the resulting gadget. Gao's construction is recovered by taking $J_0=J_{\mathrm H}=J_{\mathrm V}=X_L$. Our first main result formalizes this heterogeneous refinement of Gao's recursive gadget construction.

\begin{theorem}[Heterogeneous refinement of Gao's product]
\label{thm:heteroGao}
Let $\mathcal G_L$ and $\mathcal G_R$ be Gao gadgets in graphs $G_L$ and $G_R$, with profiles
\begin{align}
\pi(\mathcal G_L)
&=(a_L,t_L,s_L,o_L,h_L,v_L), \quad 
\pi(\mathcal G_R)
=(a_R,t_R,s_R,o_R,h_R,v_R).
\end{align}
Apply Gao's product lemma to $\mathcal G_L$ and $\mathcal G_R$, and let $I_{LR}$ denote the resulting product code. Retain Gao's selected propagated private pairs and the complementary transversals
\begin{align}
P_{LR}^{\mathrm H}
&=
(P_L^{\mathrm H}\times X_R^0)
\cup
(X_L^0\times P_R^{\mathrm V}), \nonumber\\
P_{LR}^{\mathrm V}
&=
(P_L^{\mathrm V}\times X_R^0)
\cup
(X_L^0\times P_R^{\mathrm H}).
\end{align}

Let $J_0,J_{\mathrm H},J_{\mathrm V}\subseteq V(G_L)$ be independent sets, with $J_0$ additionally satisfying
\begin{align}
J_0
\cap
N_{G_L}(P_L^{\mathrm H})
\cap
N_{G_L}(P_L^{\mathrm V})
=
\varnothing.
\end{align}
Define the decomposition of $J_0$ relative to the two left-hand transversals by
\begin{align}
J_0^0
&=
J_0\setminus
\left(
N_{G_L}(P_L^{\mathrm H})
\cup
N_{G_L}(P_L^{\mathrm V})
\right), \nonumber\\
J_0^{\mathrm H}
&=
J_0\cap N_{G_L}(P_L^{\mathrm H}), \nonumber\\
J_0^{\mathrm V}
&=
J_0\cap N_{G_L}(P_L^{\mathrm V}),
\end{align}
so that
$J_0
=
J_0^0
\mathbin{\dot\cup}
J_0^{\mathrm H}
\mathbin{\dot\cup}
J_0^{\mathrm V}$.
Write
\begin{align}
j_0&=|J_0|, \qquad
o_0=|J_0^0|, \qquad
h_0=|J_0^{\mathrm H}|, \qquad
v_0=|J_0^{\mathrm V}|, \nonumber\\
j_{\mathrm H}
&=|J_{\mathrm H}|, \qquad
j_{\mathrm V}=|J_{\mathrm V}|, \nonumber\\
q_{\mathrm H}
&=
\left|J_{\mathrm H}\setminus N_{G_L}(X_L^0)\right|,
\qquad
q_{\mathrm V}
=
\left|J_{\mathrm V}\setminus N_{G_L}(X_L^0)\right|.
\end{align}

Define the heterogeneous auxiliary set
\begin{align}
X_{LR}^{\mathrm{het}}
=
(J_0\times X_R^0)
\mathbin{\dot\cup}
(J_{\mathrm H}\times X_R^{\mathrm H})
\mathbin{\dot\cup}
(J_{\mathrm V}\times X_R^{\mathrm V}).
\end{align}
Then $X_{LR}^{\mathrm{het}}$ is an independent set and satisfies
\begin{align}
X_{LR}^{\mathrm{het}}
\cap
N(P_{LR}^{\mathrm H})
\cap
N(P_{LR}^{\mathrm V})
=
\varnothing.
\end{align}
Consequently, the unchanged code $I_{LR}$, together with Gao's selected propagated private pairs and transversals and the auxiliary set $X_{LR}^{\mathrm{het}}$, forms a Gao gadget $\mathcal G_{LR}^{\mathrm{het}}$ in $G_L\boxtimes G_R$.
The code size and the selected private-pair count are unchanged from Gao's product:
\begin{align}
a_{LR}
&=
(a_L-t_L)(a_R-t_R)+t_Ls_R+s_Lt_R, \nonumber\\
t_{LR}
&=
t_Lo_R+o_Lt_R.
\end{align}
The parameters of the heterogeneous auxiliary set are
\begin{align}
s_{LR}^{\mathrm{het}}
&=
j_0o_R+j_{\mathrm H}h_R+j_{\mathrm V}v_R, \nonumber\\
o_{LR}^{\mathrm{het}}
&=
o_0o_R+q_{\mathrm H}h_R+q_{\mathrm V}v_R, \nonumber\\
h_{LR}^{\mathrm{het}}
&=
h_0o_R+(j_{\mathrm V}-q_{\mathrm V})v_R, \nonumber\\
v_{LR}^{\mathrm{het}}
&=
v_0o_R+(j_{\mathrm H}-q_{\mathrm H})h_R.
\end{align}
Hence $\pi(\mathcal G_{LR}^{\mathrm{het}})
=
\left(
a_{LR},
t_{LR},
s_{LR}^{\mathrm{het}},
o_{LR}^{\mathrm{het}},
h_{LR}^{\mathrm{het}},
v_{LR}^{\mathrm{het}}
\right)$.
\end{theorem}

\begin{proof}
By Gao's product lemma, $I_{LR}$ is an independent set, and the selected
propagated private pairs and the transversals $P_{LR}^{\mathrm H}$ and
$P_{LR}^{\mathrm V}$ are valid. Since these objects are unchanged, it
remains only to verify that $X_{LR}^{\mathrm{het}}$ is a valid auxiliary
independent set and to determine its decomposition relative to the two
propagated transversals.

We first show that $X_{LR}^{\mathrm{het}}$ is independent. Each of the
three pieces $J_0\times X_R^0$, $J_{\mathrm H}\times X_R^{\mathrm H}$,
and $J_{\mathrm V}\times X_R^{\mathrm V}$ is independent, since both
factors in each Cartesian product are independent. Moreover, the three
right-hand classes $X_R^0$, $X_R^{\mathrm H}$, and $X_R^{\mathrm V}$
are disjoint subsets of the independent set $X_R$. Hence points belonging
to different pieces are already nonconfusable in the right coordinate.
Therefore $X_{LR}^{\mathrm{het}}$ is independent.

We next determine how points of $X_{LR}^{\mathrm{het}}$ interact with the
propagated transversals. Consider first a point
$(u,x)\in J_0\times X_R^0$. Since $x\in X_R^0$, it is nonconfusable
with every point of both $P_R^{\mathrm H}$ and $P_R^{\mathrm V}$.
Consequently, the components $X_L^0\times P_R^{\mathrm V}$ and
$X_L^0\times P_R^{\mathrm H}$ create no interaction. On the other hand,
because $x\in X_R^0$, the point $(u,x)$ is confusable with some point of
$P_L^{\mathrm H}\times X_R^0$ if and only if
$u\in N_{G_L}(P_L^{\mathrm H})$, and similarly it is confusable with
some point of $P_L^{\mathrm V}\times X_R^0$ if and only if
$u\in N_{G_L}(P_L^{\mathrm V})$. Since $J_0$ satisfies
\begin{align}
J_0
\cap
N_{G_L}(P_L^{\mathrm H})
\cap
N_{G_L}(P_L^{\mathrm V})
=
\varnothing,
\end{align}
no point in $J_0\times X_R^0$ is confusable with both propagated
transversals. More precisely, points in
$J_0^0\times X_R^0$, $J_0^{\mathrm H}\times X_R^0$, and
$J_0^{\mathrm V}\times X_R^0$ belong respectively to the new classes
$X_{LR}^0$, $X_{LR}^{\mathrm H}$, and $X_{LR}^{\mathrm V}$.

Now consider a point
$(u,x)\in J_{\mathrm H}\times X_R^{\mathrm H}$. Since
$x\in X_R^{\mathrm H}\subseteq X_R$ and $X_R$ is independent,
$x$ is nonconfusable with every point of $X_R^0$. Thus $(u,x)$ cannot
be confusable with any point of either
$P_L^{\mathrm H}\times X_R^0$ or
$P_L^{\mathrm V}\times X_R^0$. Furthermore,
$x\in X_R^{\mathrm H}$ means that $x$ is confusable with at least one
point of $P_R^{\mathrm H}$ and with no point of
$P_R^{\mathrm V}$. Hence $(u,x)$ cannot be confusable with any point of
$X_L^0\times P_R^{\mathrm V}\subseteq P_{LR}^{\mathrm H}$.
Its only possible interaction is with
$X_L^0\times P_R^{\mathrm H}\subseteq P_{LR}^{\mathrm V}$, and such an
interaction occurs if and only if
$u\in N_{G_L}(X_L^0)$. Therefore the $q_{\mathrm H}$ vertices of
$J_{\mathrm H}\setminus N_{G_L}(X_L^0)$ give points confusable with
neither propagated transversal, while the remaining
$j_{\mathrm H}-q_{\mathrm H}$ vertices give points confusable only with
$P_{LR}^{\mathrm V}$.

The case
$(u,x)\in J_{\mathrm V}\times X_R^{\mathrm V}$ is symmetric. Since
$x\in X_R^{\mathrm V}$, it is nonconfusable with every point of
$X_R^0$, is confusable with at least one point of
$P_R^{\mathrm V}$, and is nonconfusable with every point of
$P_R^{\mathrm H}$. Hence its only possible interaction is with
$X_L^0\times P_R^{\mathrm V}\subseteq P_{LR}^{\mathrm H}$, and such an
interaction occurs if and only if
$u\in N_{G_L}(X_L^0)$. Thus the $q_{\mathrm V}$ vertices of
$J_{\mathrm V}\setminus N_{G_L}(X_L^0)$ contribute to the new class
$X_{LR}^0$, while the remaining $j_{\mathrm V}-q_{\mathrm V}$ vertices
contribute to $X_{LR}^{\mathrm H}$.

It follows that no point of $X_{LR}^{\mathrm{het}}$ is confusable with
both propagated transversals, and hence
\begin{align}
X_{LR}^{\mathrm{het}}
\cap
N_{G_L\boxtimes G_R}(P_{LR}^{\mathrm H})
\cap
N_{G_L\boxtimes G_R}(P_{LR}^{\mathrm V})
=
\varnothing.
\end{align}
Thus $X_{LR}^{\mathrm{het}}$ is a valid auxiliary set, and together with
Gao's unchanged product code, selected private pairs, and transversals it
defines a Gao gadget.

More explicitly, the three auxiliary classes are
\begin{align}
X_{LR}^0
&=
(J_0^0\times X_R^0)
\mathbin{\dot\cup}
\left(
(J_{\mathrm H}\setminus N_{G_L}(X_L^0))
\times X_R^{\mathrm H}
\right)
\mathbin{\dot\cup}
\left(
(J_{\mathrm V}\setminus N_{G_L}(X_L^0))
\times X_R^{\mathrm V}
\right),
\nonumber\\
X_{LR}^{\mathrm H}
&=
(J_0^{\mathrm H}\times X_R^0)
\mathbin{\dot\cup}
\left(
(J_{\mathrm V}\cap N_{G_L}(X_L^0))
\times X_R^{\mathrm V}
\right),
\nonumber\\
X_{LR}^{\mathrm V}
&=
(J_0^{\mathrm V}\times X_R^0)
\mathbin{\dot\cup}
\left(
(J_{\mathrm H}\cap N_{G_L}(X_L^0))
\times X_R^{\mathrm H}
\right).
\end{align}
Taking cardinalities gives
\begin{align}
s_{LR}^{\mathrm{het}}
&=
j_0o_R+j_{\mathrm H}h_R+j_{\mathrm V}v_R,
\nonumber\\
o_{LR}^{\mathrm{het}}
&=
o_0o_R+q_{\mathrm H}h_R+q_{\mathrm V}v_R,
\nonumber\\
h_{LR}^{\mathrm{het}}
&=
h_0o_R+(j_{\mathrm V}-q_{\mathrm V})v_R,
\nonumber\\
v_{LR}^{\mathrm{het}}
&=
v_0o_R+(j_{\mathrm H}-q_{\mathrm H})h_R.
\end{align}
Finally, since neither the product code nor the selected propagated private
pairs have been changed, their parameters remain
\begin{align}
a_{LR}
&=
(a_L-t_L)(a_R-t_R)+t_Ls_R+s_Lt_R,
\nonumber\\
t_{LR}
&=
t_Lo_R+o_Lt_R.
\end{align}
These are precisely the claimed parameters of
$\mathcal G_{LR}^{\mathrm{het}}$. Taking $J_0=J_{\mathrm H}=J_{\mathrm V}=X_L$ recovers Gao's original auxiliary set $X_L\times X_R$ and the propagation formulas in Lemma~\ref{lemma:GaoProduct}.
\end{proof}

\begin{remark}[\textbf{Compact notation for product constructions}]
For the remainder of the paper, we write
\begin{align}
{\color{red}\operatorname{Gao}}(\mathcal{G}_L,\mathcal{G}_R)
[X_L;X_R],
\end{align}
for Gao's product of the gadgets $\mathcal{G}_L$ and $\mathcal{G}_R$ using the left and right auxiliary sets $X_L$ and $X_R$, respectively. Its output auxiliary set is $X_L\times X_R$.
For the heterogeneous refinement presented in Theorem~\ref{thm:heteroGao}, we write
\begin{align}
{\color{blue}\operatorname{HetGao}}(\mathcal{G}_L,\mathcal{G}_R)
\bigl[(J_0,J_{\mathrm H},J_{\mathrm V});X_R\bigr],
\end{align}
where $X_R=X_R^0\mathbin{\dot\cup}X_R^{\mathrm H}\mathbin{\dot\cup}X_R^{\mathrm V}$ is the fixed right-hand auxiliary set, and $J_0,J_{\mathrm H},J_{\mathrm V}\subseteq V(G_L)$ are the three left-hand codebooks used over $X_R^0,X_R^{\mathrm H},X_R^{\mathrm V}$, respectively. The output auxiliary set is $(J_0\times X_R^0)\mathbin{\dot\cup}(J_{\mathrm H}\times X_R^{\mathrm H})\mathbin{\dot\cup}(J_{\mathrm V}\times X_R^{\mathrm V})$. The main product code, propagated private pairs, and propagated transversals are exactly those of the ordinary Gao product formed using $X_L$ and $X_R$; the sets $J_0,J_{\mathrm H},J_{\mathrm V}$ enter only in the choice of the output auxiliary set, and hence in the parameters $s,o,h,v$ of the resulting gadget. Gao's ordinary product is recovered by taking $J_0=J_{\mathrm H}=J_{\mathrm V}=X_L$.
\end{remark}

\begin{remark}[\textbf{Two orientations of heterogeneity}]
The heterogeneous refinement can be applied in two orientations. In the formulation of Theorem~\ref{thm:heteroGao}, the right-hand auxiliary set $X_R$ is decomposed into $X_R^0$, $X_R^{\mathrm H}$, and $X_R^{\mathrm V}$, and three possibly different left-hand codebooks $J_0$, $J_{\mathrm H}$, and $J_{\mathrm V}$ are used over these classes. Alternatively, one may exchange the roles of the two input gadgets: decompose $X_L$ into $X_L^0$, $X_L^{\mathrm H}$, and $X_L^{\mathrm V}$ and choose three corresponding right-hand codebooks, with the codebook used over $X_L^0$ satisfying the auxiliary-set condition with respect to $P_R^{\mathrm H}$ and $P_R^{\mathrm V}$. Equivalently, this second orientation is obtained by applying Theorem~\ref{thm:heteroGao} to $(\mathcal G_R,\mathcal G_L)$ and then swapping the two coordinate blocks. The two orientations need not yield the same propagated profile, and both can therefore be useful in recursive optimization. 
\end{remark}

\begin{remark}[\textbf{On two-sided heterogeneity}]\label{remark:twosided}
Theorem~\ref{thm:heteroGao} varies the left-hand codebook while retaining the fixed right-hand auxiliary set $X_R$. This fixed side guarantees that the three product blocks are mutually separated, since $X_R^0$, $X_R^{\mathrm H}$, and $X_R^{\mathrm V}$ are disjoint subsets of the same independent set $X_R$. If unrelated codebooks were allowed on both sides, this automatic separation would no longer hold, and additional compatibility conditions would be needed to ensure that every pair of product blocks is separated in at least one coordinate. Thus, a genuinely two-sided heterogeneous construction would require additional coordination beyond the hypotheses of Theorem~\ref{thm:heteroGao}. We leave the development of such a two-sided extension, including the necessary compatibility conditions, for future work.
\end{remark}

\subsection{Warmup: Improving Gao's Recursion for $C_7$ via Heterogeneity}\label{sec:heterogeneousGaoforC7}

The profile of a Gao gadget should be viewed as recording the information that the gadget carries forward to later stages of the recursion. The parameter $a$ records the size of the current main independent set, while $t$ records private-pair structure that can create additional main-code words in subsequent products. The auxiliary set is further divided into the three parts counted by
$o$, $h$, and $v$ because these parts behave differently under later product constructions. In particular, the $\mathrm H$- and $\mathrm V$-parts already satisfy one of the two relevant avoidance conditions, and hence admit greater freedom in Theorem~\ref{thm:heteroGao}. Thus, changing the auxiliary profile matters not because any one of $o$, $h$, or $v$ is intrinsically preferable, but because it changes the choices available at subsequent stages of the recursion. Leveraging this distinction, we next show how Theorem~\ref{thm:heteroGao} can be inserted into Gao's recursive framework to obtain a stronger sequence of gadgets. As shown in Fig.~\ref{fig:gao-vs-het}, the heterogeneous construction not only replaces selected Gao products by more flexible ones, but also makes a different recursion tree advantageous. We will see that these changes propagate through the recursion and ultimately yield a larger main independent set in dimension $d=200$, improving the lower bounds obtained by Gao~\cite{Gao2026} and by the first refinement of Buys, Polak, and Zuiddam (BPZ)~\cite{BuysPolakZuiddam2026}. The purpose of this subsection is to give a transparent first application of Theorem~\ref{thm:heteroGao} within Gao's binary recursion. We do not attempt to optimize over all admissible codebook choices or all binary product trees; the construction below is chosen to make the heterogeneous mechanism and its recursive effect explicit.

We begin with the strengthened five-dimensional ($d=5$) Gao gadget used by BPZ~\cite{BuysPolakZuiddam2026}. In the notation of Definition~\ref{def:GaoGadget}, its profile is
\begin{align}
\pi(\mathcal G_5)
=
(a_5,t_5,s_5,o_5,h_5,v_5)
=
(367,8,367,322,26,19),
\end{align}
where $a_5$ is the size of the main independent set, $t_5$ is the number of selected private pairs, $s_5$ is the size of the auxiliary independent set, and $o_5$, $h_5$, and $v_5$ are the sizes of its neutral, $\mathrm H$-side, and $\mathrm V$-side parts, respectively.

\vspace{5pt}
\noindent \textbf{\underline{Step $1$}}. Applying Gao's product (Lemma~\ref{lemma:GaoProduct}), to two copies of $\mathcal G_5$,
\begin{align}
\mathcal G_{10}
=
{\color{red}\operatorname{Gao}}(\mathcal G_5,\mathcal G_5)
[X_5;X_5],
\end{align}
gives the ten-dimensional gadget with profile
\begin{align}
\pi(\mathcal G_{10})
=
(a_{10},t_{10},s_{10},o_{10},h_{10},v_{10})
=
(134753,5152,134689,105709,14490,14490).
\end{align}
Let $I_{10}$ denote its $a_{10}=134753$-word main code, and let $X_{10}=X_5\times X_5$ denote its $s_{10}=134689$-word auxiliary set.

\begin{figure}[t]
\centering
\begin{minipage}[t]{0.45\textwidth}
\centering
\textbf{Gao/BPZ recursion}\\[2.0em]
\fbox{%
\parbox{0.7\linewidth}{%
\footnotesize
\begin{align*}
\mathcal G_{10}
&={\color{red}\operatorname{Gao}}(\mathcal G_5,\mathcal G_5),\\
\mathcal G_{15}
&={\color{red}\operatorname{Gao}}(\mathcal G_5,\mathcal G_{10}),\\
\mathcal G_{25}
&={\color{red}\operatorname{Gao}}(\mathcal G_{10},\mathcal G_{15}),\\
\mathcal G_{50}
&={\color{red}\operatorname{Gao}}(\mathcal G_{25},\mathcal G_{25}),\\
\mathcal G_{100}
&={\color{red}\operatorname{Gao}}(\mathcal G_{50},\mathcal G_{50}),\\
\mathcal G_{200}
&={\color{red}\operatorname{Gao}}(\mathcal G_{100},\mathcal G_{100}).
\end{align*}
}%
}
\end{minipage}
\hfill
\begin{minipage}[t]{0.47\textwidth}
\centering
\textbf{Heterogeneous recursion}\\[0.2em]
\fbox{%
\parbox{0.92\linewidth}{%
\footnotesize
\begin{align*}
(\text{Step $1$})~~\mathcal G_{10}
&={\color{red}\operatorname{Gao}}(\mathcal G_5,\mathcal G_5),\\
(\text{Step $2$})~~\mathcal G_{15}^{\mathrm{het}}
&={\color{blue}\operatorname{HetGao}}(\mathcal G_{10},\mathcal G_5),\\
(\text{Step $3$})~~\mathcal G_{25}^{\mathrm{het}}
&={\color{blue}\operatorname{HetGao}}(\mathcal G_{15}^{\mathrm{het}},
\mathcal G_{10}),\\
(\text{Step $4$})~~\mathcal G_{40}^{\mathrm{het}}
&={\color{blue}\operatorname{HetGao}}(\mathcal G_{15}^{\mathrm{het}},
\mathcal G_{25}^{\mathrm{het}}),\\
(\text{Step $5$})~~\mathcal G_{30}^{\mathrm{het}}
&={\color{blue}\operatorname{HetGao}}(\mathcal G_{15}^{\mathrm{het}},
\mathcal G_{15}^{\mathrm{het}}),\\
(\text{Step $6$})~~\mathcal G_{60}^{\mathrm{het}}
&={\color{red}\operatorname{Gao}}(\mathcal G_{30}^{\mathrm{het}},
\mathcal G_{30}^{\mathrm{het}}),\\
(\text{Step $7$})~~\mathcal G_{100}^{\mathrm{het}}
&={\color{blue}\operatorname{HetGao}}(\mathcal G_{60}^{\mathrm{het}},
\mathcal G_{40}^{\mathrm{het}}),\\
(\text{Step $8$})~~\mathcal G_{200}^{\mathrm{het}}
&={\color{red}\operatorname{Gao}}(\mathcal G_{100}^{\mathrm{het}},
\mathcal G_{100}^{\mathrm{het}}).
\end{align*}
}%
}
\end{minipage}
\caption{Comparison of Gao/BPZ \cite{Gao2026, BuysPolakZuiddam2026} recursion (described in Section~\ref{subsec:GaoproductandC7}) and the heterogeneous recursion presented in Section~\ref{sec:heterogeneousGaoforC7}. The heterogeneous construction changes both selected product rules and the recursion tree: the improved fifteen-dimensional gadget is reused to create two branches, which are combined at dimension $100$ before the final Gao product. For readability, the bracketed auxiliary data are suppressed in the schematic and specified in the text below.}
\label{fig:gao-vs-het}
\vspace{-10pt}
\end{figure}
\vspace{5pt}
\noindent \textbf{\underline{Step $2$}}. The first opportunity to exploit Theorem~\ref{thm:heteroGao} occurs in the next step, when $\mathcal G_{10}$ is combined with $\mathcal G_5$. In Gao's product, the same left-hand auxiliary set $X_{10}$ is used throughout the auxiliary set of $\mathcal G_5$. Theorem~\ref{thm:heteroGao} allows us instead to use the decomposition
$X_5
=
X_5^0~\dot\cup~X_5^{\mathrm H}~\dot\cup~X_5^{\mathrm V}$,
and make different choices on its three parts. We retain $X_{10}$ on the neutral part $X_5^0$, while on the $\mathrm H$- and $\mathrm V$-parts we may use a full $a_{10}$-word independent set. The independent set $J^+\subseteq V(C_7^{\boxtimes 10})$ is obtained from the automorphic image $(T\times T)(I_{10})$, followed by eight exchanges along selected private pairs, where
\begin{align}
T(w_0,w_1,w_2,w_3,w_4)
=
(2-w_1,w_3,w_0,2-w_2,w_4)
\pmod 7.
\end{align}
The eight exchanges are chosen so that the resulting set remains independent, and it satisfies
\begin{align}
|J^+|
=
134753,
\qquad
q_{10}
:=
\left|J^+\setminus N(X_{10}^0)\right|
=
27488.
\end{align}

The immediate effect on the auxiliary-set size is easy to see. The $o_5=322$ neutral points still use the $s_{10}=134689$ words of $X_{10}$, whereas the $h_5+v_5=45$ one-sided points may use all $a_{10}=134753$ words of $J^+$. Hence
\begin{align}
s_{15}=
s_{10}o_5+a_{10}(h_5+v_5)=
134689\cdot322+134753\cdot45
=
49433743.
\end{align}
Compared with using $X_{10}$ on all three parts, the increase is
\begin{align}
(a_{10}-s_{10})(h_5+v_5)
=
64\cdot45
=
2880.
\end{align}
More importantly for what follows, the value $q_{10}$ controls how these choices are distributed among the neutral, $\mathrm H$-side, and $\mathrm V$-side parts of the resulting gadget, and hence affects the information available to the next stages of the recursion.

Denote the gadget obtained from this application of Theorem~\ref{thm:heteroGao} by
\begin{align}
\mathcal G_{15}^{\mathrm{het}}
=
{\color{blue}\operatorname{HetGao}}(\mathcal G_{10},\mathcal G_5)
\bigl[(X_{10},J^+,J^+);X_5\bigr].
\end{align}
Its profile is
\begin{align}
\pi(\mathcal G_{15}^{\mathrm{het}})
=
(49495055,2504616,49433743,35275258,6703815,7454670).
\end{align}
Let $X_{15}$ denote the auxiliary set of $\mathcal G_{15}^{\mathrm{het}}$. Notice that the size $a_{15}=49495055$ of the main independent set is the same as in the corresponding ordinary Gao product. The gain at this stage lies instead in the auxiliary structure carried forward by $\mathcal G_{15}^{\mathrm{het}}$. Since the output of Theorem~\ref{thm:heteroGao} is itself again a Gao gadget, this modified structure can be exploited in subsequent applications of the theorem.

Let $I_{15}$ denote the main code of $\mathcal G_{15}^{\mathrm{het}}$ and let $X_{15}^0$ denote the neutral part of its auxiliary set. To use $\mathcal G_{15}^{\mathrm{het}}$ as the left gadget in later heterogeneous steps, define
\begin{align}
J_{15}
=
(T\times T\times T)(I_{15}),
\end{align}
where $T$ is applied independently to the three five-coordinate blocks. Since $T\times T\times T$ is an automorphism of $C_7^{\boxtimes 15}$, $J_{15}$ is an independent set with
\begin{align}
|J_{15}|
=
|I_{15}|
=
49495055.
\end{align}
For this placement we obtain the certified count
\begin{align}
q_{15}
:=
\left|J_{15}\setminus N(X_{15}^0)\right|
=
12872271.
\end{align}
No exchanges are used in constructing $J_{15}$; it is a pure automorphic image of the main code. Thus $J_{15}$ is immediately admissible for the $\mathrm H$- and $\mathrm V$-parts in Theorem~\ref{thm:heteroGao}, while the certified value of $q_{15}$ determines the auxiliary profiles produced in the subsequent heterogeneous products.

The availability of this fifteen-dimensional choice changes the recursive possibilities. Rather than continuing along a single chain, we reuse $\mathcal G_{15}^{\mathrm{het}}$ to construct two different branches, as shown in Fig.~\ref{fig:gao-vs-het}. In each such application, the left-hand codebook triple is $(X_{15},J_{15},J_{15})$. The choice $J_{0}=X_{15}$ satisfies the stronger neutral-set condition in Theorem~\ref{thm:heteroGao}, while $J_{15}$ is independent and therefore admissible for the $\mathrm H$- and $\mathrm V$-parts.

\vspace{5pt}
\noindent \textbf{\underline{Steps $3$ and $4$}}.
For the first branch, we combine $\mathcal G_{15}^{\mathrm{het}}$ with $\mathcal G_{10}$ to obtain
\begin{align}
\mathcal G_{25}^{\mathrm{het}}
=
{\color{blue}\operatorname{HetGao}}
(\mathcal G_{15}^{\mathrm{het}},\mathcal G_{10})
\bigl[(X_{15},J_{15},J_{15});X_{10}\bigr],
\end{align}
with profile
\begin{align}
\pi(\mathcal G_{25}^{\mathrm{het}})
=
(&6682034753199,\,
446498581960,\,
6659958232687,\nonumber\\
&4101950661502,\,
1239317719995,\,
1318689851190).
\end{align}
Let $X_{25}$ denote the auxiliary set of $\mathcal G_{25}^{\mathrm{het}}$. We then reuse $\mathcal G_{15}^{\mathrm{het}}$ on the left once more:
\begin{align}
\mathcal G_{40}^{\mathrm{het}}
=
{\color{blue}\operatorname{HetGao}}
(\mathcal G_{15}^{\mathrm{het}},\mathcal G_{25}^{\mathrm{het}})
\bigl[(X_{15},J_{15},J_{15});X_{25}\bigr],
\end{align}
which gives
\begin{align}
\pi(\mathcal G_{40}^{\mathrm{het}})
=
(&331763316186294443393,\,
26024163933281638912,\,
329383500225587852161,\nonumber\\
&177624734564098828651,\,
75792811956960543090,\,
75965953704528480420).
\end{align}
Let $X_{40}$ denote the auxiliary set of $\mathcal G_{40}^{\mathrm{het}}$.

\vspace{5pt}
\noindent \textbf{\underline{Steps $5$ and $6$}}.
For the second branch, we instead combine two copies of the fifteen-dimensional heterogeneous gadget:
\begin{align}
\mathcal G_{30}^{\mathrm{het}}
=
{\color{blue}\operatorname{HetGao}}
(\mathcal G_{15}^{\mathrm{het}},\mathcal G_{15}^{\mathrm{het}})
\bigl[(X_{15},J_{15},J_{15});X_{15}\bigr],
\end{align}
giving
\begin{align}
\pi(\mathcal G_{30}^{\mathrm{het}})
=
(&2455726444728097,\,
176701951181856,\,
2444563032022369,\nonumber\\
&1426595682835999,\,
509489572910550,\,
508477776275820).
\end{align}
Let $X_{30}$ denote the auxiliary set of $\mathcal G_{30}^{\mathrm{het}}$. We then apply Gao's product lemma to two copies of this gadget:
\begin{align}
\mathcal G_{60}^{\mathrm{het}}
=
{\color{red}\operatorname{Gao}}
(\mathcal G_{30}^{\mathrm{het}},\mathcal G_{30}^{\mathrm{het}})
[X_{30};X_{30}],
\end{align}
with profile
\begin{align}
\pi(\mathcal G_{60}^{\mathrm{het}})
=
(&6057870757274473350846763103809,\,
504164481409466441663744868288,\nonumber\\
&5975888417530397884926116372161,\,
3071432766295835202512524104901,\nonumber\\
&1452227825617281341206796133630,\,
1452227825617281341206796133630).
\end{align}

The same fifteen-dimensional gadget therefore contributes differently to the two branches: the first builds the $25$- and $40$-dimensional states, while the second first produces a $30$-dimensional state that is then amplified by an ordinary Gao squaring. This illustrates why a heterogeneous profile should be viewed as a recursive resource rather than only through its immediate main-code size.

\vspace{5pt}
\noindent \textbf{\underline{Step $7$}}.
The two branches are now combined at dimension $100$. At this stage we again apply Theorem~\ref{thm:heteroGao}, but in the simpler form appropriate for the penultimate step. Let $I_{60}$ and $X_{60}$ denote the main code and auxiliary set of $\mathcal G_{60}^{\mathrm{het}}$, respectively. We use $\mathcal G_{60}^{\mathrm{het}}$ as the left gadget and $\mathcal G_{40}^{\mathrm{het}}$ as the right gadget, with left-hand codebook triple $(X_{60},I_{60},I_{60})$. These choices require no additional construction: $X_{60}$ is an independent auxiliary set satisfying the neutral-set condition by the gadget axioms, while $I_{60}$ is independent. Hence
\begin{align}
\mathcal G_{100}^{\mathrm{het}}
=
{\color{blue}\operatorname{HetGao}}
(\mathcal G_{60}^{\mathrm{het}},\mathcal G_{40}^{\mathrm{het}})
\bigl[(X_{60},I_{60},I_{60});X_{40}\bigr].
\end{align}
Since the resulting $100$-dimensional gadget will be used only in the final ordinary Gao product, no additional overlap count is required. Such a count would affect only the subdivision of its auxiliary set into the neutral, $\mathrm H$-side, and $\mathrm V$-side parts, whereas the final Gao product depends only on $a_{100}$, $t_{100}$, and $s_{100}$. Theorem~\ref{thm:heteroGao} gives
\begin{align}
a_{100}
&=
2019566410046082519473049091812626784766961954468801,\\
t_{100}
&=
169483552007138542660886975007057178898284167227200,\\
s_{100}
&=
1980800582609313371231331570113696358456445857134401.
\end{align}
Let $X_{100}$ denote the auxiliary set of $\mathcal G_{100}^{\mathrm{het}}$.

\vspace{5pt}
\noindent \textbf{\underline{Step $8$}}.
For the final step, we return to Gao's product lemma and combine two copies of $\mathcal G_{100}^{\mathrm{het}}$:
\begin{align}
\mathcal G_{200}^{\mathrm{het}}
=
{\color{red}\operatorname{Gao}}
(\mathcal G_{100}^{\mathrm{het}},\mathcal G_{100}^{\mathrm{het}})
[X_{100};X_{100}].
\end{align}
By Lemma~\ref{lemma:GaoProduct}, the resulting main independent set has size
\begin{align}
a_{200}
&=
(a_{100}-t_{100})^2+2t_{100}s_{100}
\nonumber\\
&=
4094232818726419107644671962602278993363036698767319
\nonumber\\
&\qquad
682486929900807041313140911658278692029232150857601.
\end{align}
Consequently,
\begin{align}
\Shannon(C_7)
&\ge
a_{200}^{1/200}=
3.2588236744275819433344360437765093813959865800495343
\ldots .
\end{align}
This improves upon both Gao's bound and the first BPZ refinement:
\begin{align}
\underbrace{3.25878915\ldots}_{\text{Gao~}\cite{Gao2026}}
<
\underbrace{3.25880536\ldots}_{\text{BPZ~}\cite{BuysPolakZuiddam2026}}
<
\underbrace{3.25882367\ldots}_{\textit{``Warmup'' improvement}}.
\end{align}
Our heterogeneous recursion uses the same five-dimensional base gadget
as the first BPZ refinement. The improvement comes from using the additional choices of Theorem~\ref{thm:heteroGao} to produce new intermediate gadgets and then reorganizing the recursion to exploit their different profiles. This construction remains below the stronger BPZ multi-gadget bound \cite{BuysPolakZuiddamGitHub2026}
$\Shannon(C_7)\ge 3.2588279859\ldots$,
which we discuss in the next section. The purpose of the present example is to show that heterogeneity can improve not only the gadgets propagated through Gao's binary framework, but also the recursion tree through which those gadgets are most effectively combined.


\section{Heterogeneity in the BPZ Multi-Gadget Recursion}\label{subsec:BPZgeneral}
The preceding section showed that heterogeneity can strengthen Gao's binary
product framework and can also make a different recursion tree advantageous.
We now turn to the more general and recently introduced \textit{multi-gadget recursion} of Buys, Polak, and Zuiddam (BPZ) presented in their Lean formalization~\cite{BuysPolakZuiddamGitHub2026}, which allows several gadgets to be combined through more flexible combining rules and yields a stronger lower bound for $C_7$ and other odd cycle graphs. We first describe the specific BPZ recursion underlying this improvement and the additional flexibility it introduces beyond Gao's construction. We then show that the heterogeneous ideas developed in Section~\ref{sec:heterogeneousGao} can be incorporated into this BPZ recursion as well, leading to the main improvement of the paper which is presented in Theorem~\ref{thm:mainC7}.

\subsection{The BPZ Multi-Gadget Recursion}\label{subsec:BPZmultigadget}

The key generalization introduced by BPZ~\cite{BuysPolakZuiddamGitHub2026} is to retain more of the internal
structure of a Gao gadget throughout the recursion. Rather than combining
gadgets only through Gao's fixed binary product, BPZ keep track of seven
constituent sets and the separation relations among them. An
admissible combining rule then specifies which Cartesian products of these constituent sets may be assembled while preserving the required independence
and separation properties. Gao's binary product (Lemma~\ref{lemma:GaoProduct}) appears as one particular
rule within this more general framework.

\begin{definition}[Set separation]\label{def:BPZseparation}
Let $Y,Z\subseteq V(G)$. We say that $Y$ and $Z$ are \emph{separated}, and write $Y\perp_G Z$, if no vertex of $Y$ is confusable with a vertex of $Z$; that is,
\begin{align}
y\not\simeq z
\qquad
\text{for every } y\in Y,\ z\in Z.
\end{align}
\end{definition}

\begin{definition}[BPZ set labels $(\mathcal L)$ and label separation]
BPZ represent the structure propagated by a gadget using \textit{seven labels} for
its constituent sets, together with prescribed separation relations among
these labels. Specifically, let
$\mathcal L
=
\{B,N,A,D,O,H,V\}$
denote the set of seven labels. We define a symmetric relation $\perp$ on $\mathcal L$ by specifying, for
each $\lambda\in\mathcal L$, exactly those labels $\mu\in\mathcal L$ for
which $\lambda\perp\mu$:
\begin{center}
\renewcommand{\arraystretch}{1.15}
\begin{tabular}{c|c}
$\lambda$
&
$\{\,\mu\in\mathcal L:\lambda\perp\mu\,\}$
\\ \hline
$B$ & $\{O,H,V\}$ \\
$N$ & $\{A,D,O,H,V\}$ \\
$A$ & $\{N,D,H\}$ \\
$D$ & $\{N,A,V\}$ \\
$O$ & $\{B,N\}$ \\
$H$ & $\{B,N,A\}$ \\
$V$ & $\{B,N,D\}$
\end{tabular}
\end{center}
\end{definition}

\begin{definition}[BPZ seven-family representation] \label{def:BPZ7familyrep}
A \emph{seven-family representation} in a graph $G$ is a collection
\begin{align}
\mathcal{F}=\{F_\lambda:\lambda\in\mathcal{L}\}
\end{align}
such that every $F_\lambda\subseteq V(G)$ is an independent set and
$F_\lambda\perp_G F_\mu$
whenever
$\lambda\perp\mu$.
\end{definition}

\vspace{5pt}
\noindent\textbf{Converting a Gao gadget to a BPZ seven-family representation.}
Every Gao gadget naturally gives a BPZ seven-family representation as we describe next. Let
$\mathcal G$ be a Gao gadget as in Definition~\ref{def:GaoGadget}, with main
independent set $I=B\dot\cup R$, auxiliary-set decomposition
$X=X^0~\dot\cup~ X^{\mathrm H}~\dot\cup~ X^{\mathrm V}$, and complementary
transversals $P^{\mathrm H}$ and $P^{\mathrm V}$. The correspondence between
the Gao-gadget sets and the seven BPZ families is given in
Table~\ref{tab:GaoToBPZconversion}(\subref{tab:GaoToBPZfamilies}), while
Table~\ref{tab:GaoToBPZconversion}(\subref{tab:GaoToBPZseparations}) records
how the defining properties of a Gao gadget imply all eleven separation
requirements.
\begin{table}[t]
\centering
\small

\begin{subtable}[t]{0.30\textwidth}
\centering
\caption{Family correspondence.}
\label{tab:GaoToBPZfamilies}
\vspace{5.0em}
\renewcommand{\arraystretch}{1.15}
\begin{tabular}{@{}ccc@{}}
\toprule
\textbf{Family} & \textbf{Gao set} & \textbf{Size} \\
\midrule
$F_B$ & $B$ & $a-t$ \\
$F_N$ & $X^0$ & $o$ \\
$F_A$ & $X^{\mathrm V}$ & $v$ \\
$F_D$ & $X^{\mathrm H}$ & $h$ \\
$F_O$ & $R$ & $t$ \\
$F_H$ & $P^{\mathrm H}$ & $t$ \\
$F_V$ & $P^{\mathrm V}$ & $t$ \\
\bottomrule
\end{tabular}
\end{subtable}
\hfill
\begin{subtable}[t]{0.67\textwidth}
\centering
\caption{Separation justifications.}
\label{tab:GaoToBPZseparations}
\renewcommand{\arraystretch}{1.18}
\begin{tabularx}{\linewidth}{
@{}
>{\raggedright\arraybackslash}p{0.29\linewidth}
>{\raggedright\arraybackslash}X
@{}}
\toprule
\textbf{Required separation} & \textbf{Gao-gadget property} \\
\midrule

$B\perp O$
&
$B$ and $R$ are disjoint subsets of the independent set $I$.
\\

$B\perp H,\ B\perp V$
&
The private-pair property and the independence of $I$ imply that $B$ is
separated from every private-pair endpoint; both transversals consist of
such endpoints.
\\

$\begin{gathered}
N\perp A,\ N\perp D,\\
A\perp D
\end{gathered}$
&
The sets $X^0$, $X^{\mathrm V}$, and $X^{\mathrm H}$ are disjoint subsets
of the independent set $X$.
\\

$N\perp H,\ N\perp V$
&
By definition, $X^0$ is separated from both $P^{\mathrm H}$ and
$P^{\mathrm V}$.
\\

$N\perp O$
&
Every center in $R$ belongs to one of the two transversals, while $X^0$ is
separated from both transversals.
\\

$A\perp H,\ D\perp V$
&
The condition
$X\cap N(P^{\mathrm H})\cap N(P^{\mathrm V})=\emptyset$
implies that $X^{\mathrm V}$ is separated from $P^{\mathrm H}$ and that
$X^{\mathrm H}$ is separated from $P^{\mathrm V}$.
\\

\bottomrule
\end{tabularx}
\end{subtable}

\caption{Conversion of a Gao gadget into a BPZ seven-family representation.
Part~(\subref{tab:GaoToBPZfamilies}) identifies the seven constituent sets and
their cardinalities, while Part~(\subref{tab:GaoToBPZseparations}) shows how
the defining properties of a Gao gadget imply the eleven separation
requirements of Definition~\ref{def:BPZ7familyrep}.}
\label{tab:GaoToBPZconversion}
\vspace{-15pt}
\end{table}
Each of the seven sets in
Table~\ref{tab:GaoToBPZconversion}(\subref{tab:GaoToBPZfamilies}) is
independent: $B$ and $R$ are subsets of the independent set $I$;
$X^0$, $X^{\mathrm H}$, and $X^{\mathrm V}$ are subsets of the independent
set $X$; and the independence of $P^{\mathrm H}$ and $P^{\mathrm V}$ is
part of the Gao-gadget definition. Together with the separation properties
in Table~\ref{tab:GaoToBPZconversion}(\subref{tab:GaoToBPZseparations}), this
shows that the seven sets form a BPZ seven-family representation. For a
seven-family representation $\mathcal F$, we refer to
$(|F_B|,|F_N|,|F_A|,|F_D|,|F_O|,|F_H|,|F_V|)$ as its cardinality vector, always
recorded in the order $(B,N,A,D,O,H,V)$. We denote the resulting Gao-to-BPZ
cardinality map by
\begin{align}\label{eq:Gao-to-BPZ}
\Phi(a,t,s,o,h,v)
:=
(a-t,o,v,h,t,t,t),
\end{align}
where the seven entries on the right are recorded in the fixed order
$(B,N,A,D,O,H,V)$.
Thus the $A$- and $D$-coordinates of the seven-family cardinality vector are
$v$ and $h$, respectively; equivalently,
\begin{align}
|F_A|=v,
\qquad
|F_D|=h.
\end{align}
The BPZ labels $H$ and $V$ refer to the transversals
$P^{\mathrm H}$ and $P^{\mathrm V}$, respectively, and should not be
confused with the auxiliary sets $X^{\mathrm H}$ and $X^{\mathrm V}$.

\vspace{5pt}
\noindent\textbf{Reversing the orientation.}
Given a Gao gadget $\mathcal G$, let $\sigma\mathcal G$ denote the same
gadget with its two complementary transversals interchanged:
$P^{\mathrm H}
\longleftrightarrow
P^{\mathrm V}$.
The main independent set, the private pairs, and the auxiliary set are left
unchanged. Since the neutral part depends only on the union of the two
transversal neighborhoods, $X^0$ is unchanged, while the two one-sided
parts are exchanged:
$X^{\mathrm H}
\longleftrightarrow
X^{\mathrm V}$.
Consequently, 
\begin{align}
\text{if~~}\pi(\mathcal G)
=
(a,t,s,o,h,v),
\text{~~then~~}
\pi(\sigma\mathcal G)
=
(a,t,s,o,v,h).
\end{align}
Under the seven-family correspondence, reversing the orientation exchanges
\begin{align}
F_A&\longleftrightarrow F_D,
&
F_H&\longleftrightarrow F_V,
\end{align}
while leaving $F_B$, $F_N$, and $F_O$ unchanged. Hence, for a seven-family
cardinality vector recorded in the order $(B,N,A,D,O,H,V)$,
\begin{align}
\sigma(B,N,A,D,O,H,V)
=
(B,N,D,A,O,V,H).
\end{align}
This is compatible with the conversion in \eqref{eq:Gao-to-BPZ}. Although
$|F_H|=|F_V|=t$ for every representation arising from a Gao gadget, the
underlying sets $F_H$ and $F_V$ are still interchanged by $\sigma$.

\begin{definition}[Admissible combining rule]
\label{def:BPZadmissibleRule}
Let $m\geq 2$. An $m$-ary combining rule is a collection
\begin{align}
\mathcal T
=
\{T_\lambda:\lambda\in\mathcal L\},
\qquad
T_\lambda\subseteq\mathcal L^m,
\end{align}
where each $T_\lambda$ is a set of ordered $m$-tuples of labels. The rule is
called \emph{admissible} if the following two conditions hold.
\begin{enumerate}
\item[(i)] For every $\lambda\in\mathcal L$ and every two distinct tuples
$(\lambda_1,\ldots,\lambda_m),(\mu_1,\ldots,\mu_m)\in T_\lambda$, there
exists an index $i\in\{1,\ldots,m\}$ such that
$\lambda_i\perp\mu_i$.

\item[(ii)] Whenever $\lambda\perp\mu$, for every
$(\lambda_1,\ldots,\lambda_m)\in T_\lambda$ and
$(\mu_1,\ldots,\mu_m)\in T_\mu$, there exists an index
$i\in\{1,\ldots,m\}$ such that $\lambda_i\perp\mu_i$.
\end{enumerate}

Condition~(i) ensures that each output family is independent, while
Condition~(ii) ensures that the prescribed separation relations between
different output families are preserved.
\end{definition}

Now let
$\mathcal F^{(1)},\ldots,\mathcal F^{(m)}$ be seven-family representations
in graphs $G_1,\ldots,G_m$, respectively, and let $\mathcal T$ be an
admissible combining rule. The $i$th coordinate of each tuple in
$T_\lambda$ is paired with the $i$th input representation
$\mathcal F^{(i)}$. For each $\lambda\in\mathcal L$, define
\begin{align}
F_\lambda^{\mathrm{out}}
=
\bigcup_{(\lambda_1,\ldots,\lambda_m)\in T_\lambda}
F_{\lambda_1}^{(1)}
\times\cdots\times
F_{\lambda_m}^{(m)}.
\end{align}

\begin{proposition}[BPZ generalized combining rule]
\label{prop:BPZgeneralizedCombiningRule}
The collection
\begin{align}
\mathcal F^{\mathrm{out}}
=
\{F_\lambda^{\mathrm{out}}:\lambda\in\mathcal L\}
\end{align}
is a seven-family representation in
$G_1\boxtimes\cdots\boxtimes G_m$. Moreover, for every
$\lambda\in\mathcal L$,
\begin{align}
|F_\lambda^{\mathrm{out}}|
=
\sum_{(\lambda_1,\ldots,\lambda_m)\in T_\lambda}
\prod_{i=1}^m
|F_{\lambda_i}^{(i)}|.
\end{align}
\end{proposition}

\begin{proof}
Fix $\lambda\in\mathcal L$. For every tuple
$(\lambda_1,\ldots,\lambda_m)\in T_\lambda$, the Cartesian product
\begin{align}
F_{\lambda_1}^{(1)}
\times\cdots\times
F_{\lambda_m}^{(m)}
\end{align}
is independent because each of its factors is independent. For two distinct
tuples in $T_\lambda$, Condition~(i) gives a coordinate in which the
corresponding constituent sets are separated. The two Cartesian products are
therefore separated, and their union is independent. Hence
$F_\lambda^{\mathrm{out}}$ is independent.

Now suppose that $\lambda\perp\mu$. By Condition~(ii), every Cartesian
product contributing to $F_\lambda^{\mathrm{out}}$ is separated from every
Cartesian product contributing to $F_\mu^{\mathrm{out}}$. It follows that
\begin{align}
F_\lambda^{\mathrm{out}}
\perp_{G_1\boxtimes\cdots\boxtimes G_m}
F_\mu^{\mathrm{out}}.
\end{align}
Thus all the separation relations required in
Definition~\ref{def:BPZ7familyrep} are preserved, and
$\mathcal F^{\mathrm{out}}$ is a seven-family representation. Finally, the Cartesian products indexed by distinct tuples in $T_\lambda$
are disjoint, since they are separated in at least one coordinate.
Consequently, their cardinalities add, while the cardinality of each
Cartesian product is the product of the cardinalities of its factors. This
gives
\begin{align}
|F_\lambda^{\mathrm{out}}|
=
\sum_{(\lambda_1,\ldots,\lambda_m)\in T_\lambda}
\prod_{i=1}^m
|F_{\lambda_i}^{(i)}|.
\end{align}
\end{proof}

In particular, if $\mathcal F^{(i)}$ is a seven-family representation in
$G^{\strong d_i}$ for each $i\in\{1,\ldots,m\}$, then the resulting
seven-family representation is in
$G^{\strong(d_1+\cdots+d_m)}$.
The framework is quite general, but its additional flexibility becomes more
transparent through a few examples. 

To distinguish the components $T_\lambda$ of a generic combining rule from
the particular rules in the BPZ repository, we denote the named admissible
combining rules by $S_{m\alpha}$, where $m$ records the arity. By
Proposition~\ref{prop:BPZgeneralizedCombiningRule}, each such rule also induces
a map on seven-family cardinality vectors; we use the same symbol for the
set-level rule and its induced cardinality map. The repository also contains
terminal codes $K_{m\alpha}$. A terminal code differs from a combining rule in
that it produces a single independent set rather than a new seven-family
representation. We begin by showing that Gao's binary product is itself a
particular admissible combining rule.

\vspace{5pt}
\noindent \textit{\textbf{Example $2$: Recovering Gao's binary product as a
special case of BPZ.}}
\textit{Consider the following admissible binary combining rule, denoted by
$S_{2a}$:
\begin{align}
T_B &= \{(B,B),(H,D),(V,A),(D,V),(A,H)\}, \nonumber\\
T_N &= \{(N,N),(A,A),(A,D),(D,A),(D,D)\}, \nonumber\\
T_A &= \{(A,N),(N,D)\}, \qquad
T_D = \{(D,N),(N,A)\}, \nonumber\\
T_O &= \{(O,N),(N,O)\}, \nonumber\\
T_H &= \{(H,N),(N,V)\}, \qquad
T_V = \{(V,N),(N,H)\}.
\end{align}
Under the correspondence established above,
\begin{align}
|F_B^{(i)}|&=a_i-t_i,
&
|F_N^{(i)}|&=o_i,
&
|F_A^{(i)}|&=v_i,
&
|F_D^{(i)}|&=h_i,
\nonumber\\
|F_O^{(i)}|
&=
|F_H^{(i)}|
=
|F_V^{(i)}|
=
t_i,
\end{align}
the resulting cardinalities reproduce Gao's product formulas exactly. For
example,
\begin{align}
|F_O^{\mathrm{out}}|
&=
t_1o_2+o_1t_2
=
t_{12},\\
|F_N^{\mathrm{out}}|
&=
o_1o_2+(h_1+v_1)(h_2+v_2)
=
o_{12},\\
|F_D^{\mathrm{out}}|
&=
h_1o_2+o_1v_2
=
h_{12},\\
|F_A^{\mathrm{out}}|
&=
v_1o_2+o_1h_2
=
v_{12}.
\end{align}
Moreover, the main independent set of a Gao gadget satisfies
\begin{align}
I
=
B\dot\cup R
=
F_B\dot\cup F_O.
\end{align}
Consequently,
\begin{align}
|F_B^{\mathrm{out}}|+|F_O^{\mathrm{out}}|
&=
(a_1-t_1)(a_2-t_2)+t_1s_2+s_1t_2
\nonumber\\
&=
a_{12}.
\end{align}
Thus the ordinary product
$\operatorname{Gao}(\mathcal G_1,\mathcal G_2)[X_1;X_2]$ is recovered as the
particular admissible binary combining rule $S_{2a}$.}

\begin{table}[t]
\centering
\renewcommand{\arraystretch}{1.2}
\begin{tabularx}{0.92\textwidth}{
@{}
>{\raggedright\arraybackslash}l
>{\centering\arraybackslash}p{0.08\textwidth}
>{\raggedright\arraybackslash}X
@{}}
\toprule
\textbf{Object type}
&
\textbf{Arity}
&
\textbf{Lean-verified rules or codes}
\\
\midrule

Admissible combining rules
&
$2$
&
$S_{2a}\ (20),\qquad S_{2b}\ (17)$
\\

Admissible combining rules
&
$3$
&
$S_{3a}\ (58),\quad S_{3b}\ (45),\quad
S_{3c}\ (48),\quad S_{3d}\ (46),$
\newline
$S_{3e}\ (53),\quad S_{3f}\ (47),\quad
S_{3g}\ (47),\quad S_{3h}\ (45)$
\\

Terminal codes
&
$3$
&
$K_{3a}\ (19)$
\\

Terminal codes
&
$4$
&
$K_{4a}\ (49),\qquad K_{4b}\ (57)$
\\

\bottomrule
\end{tabularx}
\caption{The admissible combining rules and terminal codes verified in the
BPZ Lean repository~\cite{BuysPolakZuiddamGitHub2026}. The arity $(m)$ is the
number of input seven-family representations that are combined, and the number in parentheses
is the number of ordered words in the corresponding rule or terminal code.}
\label{tab:BPZ-rule-library}
\end{table}

\vspace{5pt}
\noindent \textit{\textbf{Example $3$: Combining three inputs in one step.}}
\textit{The BPZ framework also permits several seven-family representations
to be combined simultaneously. Consider the admissible ternary rule
$S_{3a}$, which contains $58$ ordered triples in total. Rather than reproduce
the complete rule, consider its $O$ component:
\begin{align}
T_O
=
\{(O,A,A),(O,A,D),(O,D,A),(O,D,D)\}.
\end{align}
For three seven-family representations
$\mathcal F^{(1)},\mathcal F^{(2)},\mathcal F^{(3)}$, this component gives
\begin{align}
F_O^{\mathrm{out}}
={}&
(F_O^{(1)}\times F_A^{(2)}\times F_A^{(3)})
\,\dot\cup\,
(F_O^{(1)}\times F_A^{(2)}\times F_D^{(3)})
\nonumber\\
&\,\dot\cup\,
(F_O^{(1)}\times F_D^{(2)}\times F_A^{(3)})
\,\dot\cup\,
(F_O^{(1)}\times F_D^{(2)}\times F_D^{(3)}),
\end{align}
and hence
\begin{align}
|F_O^{\mathrm{out}}|
=
|F_O^{(1)}|
\bigl(
|F_A^{(2)}|+|F_D^{(2)}|
\bigr)
\bigl(
|F_A^{(3)}|+|F_D^{(3)}|
\bigr).
\end{align}
The remaining six output families are obtained from the corresponding
components of $S_{3a}$ in the same manner. This illustrates the additional
freedom of the BPZ construction: several representations may be combined in
one step, with a separate collection of allowed tuples specified for each of
the seven output families.}

The version of the BPZ Lean repository \cite{BuysPolakZuiddamGitHub2026} used here contains the admissible
combining rules and terminal codes summarized in
Table~\ref{tab:BPZ-rule-library}. The number in parentheses is the total
number of ordered words in the corresponding rule or code. All BPZ combining rules are presented in Appendix~\ref{app:BPZ-rules}. The specific BPZ construction for $C_7$ described below uses the combining
rules $S_{2a}, S_{3a}, S_{3b}$,
followed by the terminal code $K_{4a}$. The remaining verified rules and
terminal codes illustrate the substantially larger design space made
available by the multi-gadget framework.

\vspace{5pt}
\paragraph{The BPZ multi-gadget recursion for $C_7$.}
We now specialize the preceding framework to the BPZ construction for
$C_7$~\cite{BuysPolakZuiddamGitHub2026}. The five-dimensional base
representation is induced by the Gao gadget
\begin{align}
\pi(\mathcal G_5)
=
(367,8,367,322,26,19).
\end{align}
Applying the map $\Phi$ in \eqref{eq:Gao-to-BPZ} gives the cardinality
vectors associated with the two orientations of the five-dimensional base
gadget:
\begin{align}
\mathbf w
&:=
\Phi\bigl(\pi(\mathcal G_5)\bigr)
=
(359,322,19,26,8,8,8),
\\
\sigma\mathbf w
&:=
\Phi\bigl(\pi(\sigma\mathcal G_5)\bigr)
=
(359,322,26,19,8,8,8).
\end{align}

\begin{figure}[t]
\centering
\includegraphics[width=\textwidth]{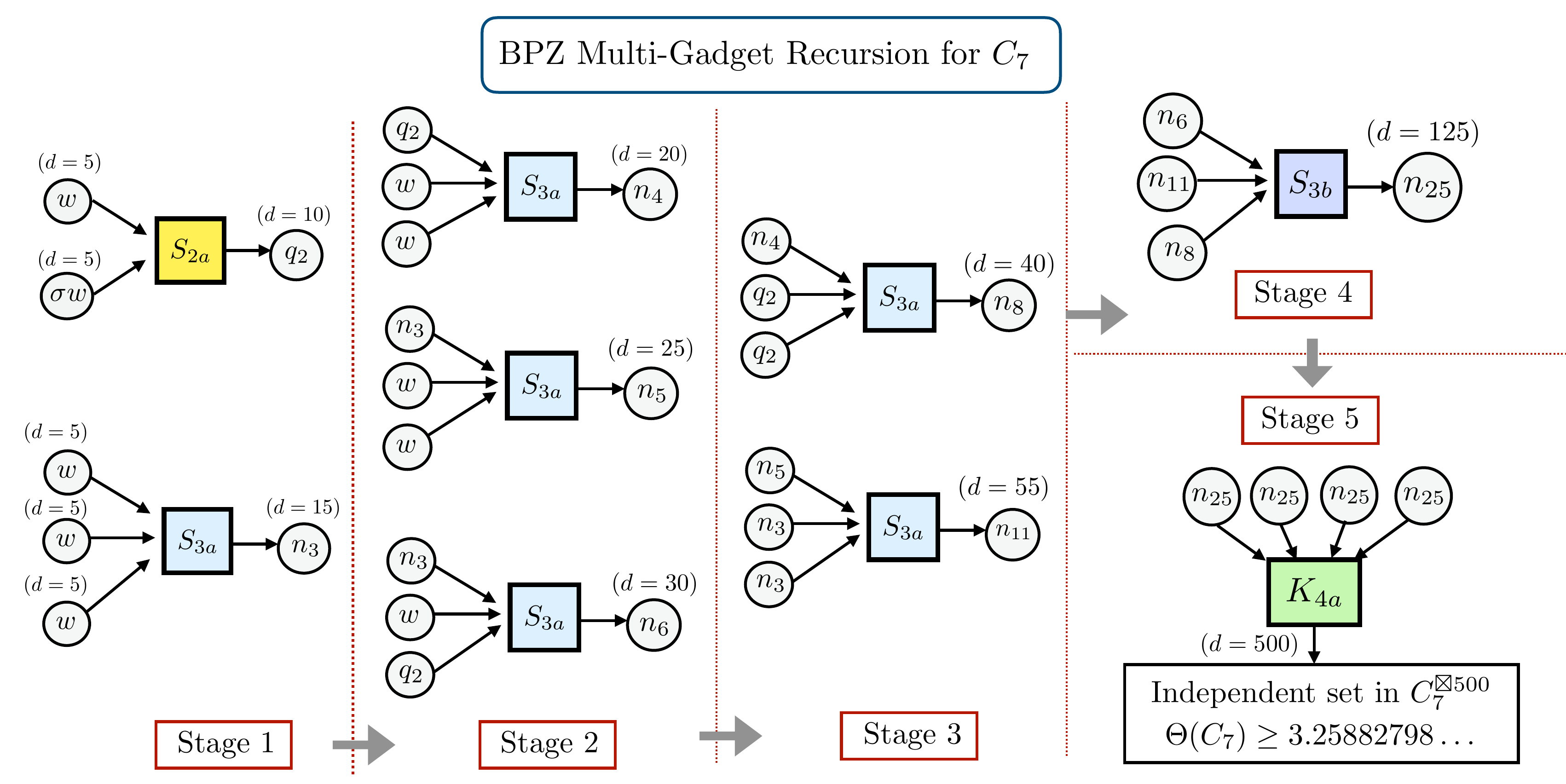}
\caption{The BPZ multi-gadget recursion for $C_7$. Starting from two
oriented five-dimensional base representations with cardinality vectors
$\mathbf w$ and $\sigma\mathbf w$, Stages $1$--$3$ use the combining rules
$S_{2a}$ and $S_{3a}$ to construct intermediate representations with
cardinality vectors $\mathbf q_2$, $\mathbf n_3$, $\mathbf n_4$,
$\mathbf n_5$, $\mathbf n_6$, $\mathbf n_8$, and $\mathbf n_{11}$. Stage $4$
applies $S_{3b}$ to obtain a $125$-dimensional representation with cardinality
vector $\mathbf n_{25}$. In Stage $5$, the terminal code $K_{4a}$ is applied
to four copies of this representation to produce an independent set in
$C_7^{\strong 500}$.}
\label{fig:BPZv2-recursion}
\end{figure}

Starting from representations with cardinality vectors $\mathbf w$ and
$\sigma\mathbf w$, BPZ repeatedly apply the combining rules $S_{2a}$ and
$S_{3a}$, followed by the ternary rule $S_{3b}$. We write $\mathbf n_r$ for
the cardinality vector of a seven-family representation constructed from $r$
copies of the five-dimensional base block. The symbol $\mathbf q_2$ denotes a
second two-block cardinality vector and is unrelated to the overlap quantities
denoted by $q$ in Theorem~\ref{thm:heteroGao}. The ordered recursion is
\begin{align}
\mathbf q_2
&=
S_{2a}(\mathbf w,\sigma\mathbf w),
&
\mathbf n_3
&=
S_{3a}(\mathbf w,\mathbf w,\mathbf w),
\nonumber\\
\mathbf n_4
&=
S_{3a}(\mathbf q_2,\mathbf w,\mathbf w),
&
\mathbf n_5
&=
S_{3a}(\mathbf n_3,\mathbf w,\mathbf w),
\nonumber\\
\mathbf n_6
&=
S_{3a}(\mathbf n_3,\mathbf w,\mathbf q_2),
&
\mathbf n_8
&=
S_{3a}(\mathbf n_4,\mathbf q_2,\mathbf q_2),
\nonumber\\
\mathbf n_{11}
&=
S_{3a}(\mathbf n_5,\mathbf n_3,\mathbf n_3),
&
\mathbf n_{25}
&=
S_{3b}(\mathbf n_6,\mathbf n_{11},\mathbf n_8).
\end{align}
The order of the inputs in each application is part of the construction. The
subscript records the number of five-dimensional base blocks used:
$\mathbf n_r$ is the cardinality vector of a representation in
$C_7^{\strong 5r}$, while $\mathbf q_2$ is the cardinality vector of a
representation in $C_7^{\strong 10}$. In particular, $\mathbf n_6$,
$\mathbf n_8$, $\mathbf n_{11}$, and $\mathbf n_{25}$ are realized by
seven-family representations in $C_7^{\strong 30}$, $C_7^{\strong 40}$,
$C_7^{\strong 55}$, and $C_7^{\strong 125}$, respectively. The complete
dependency structure is shown in Fig.~\ref{fig:BPZv2-recursion}.

To obtain the final independent set, BPZ apply the terminal code $K_{4a}$ to
four copies of a seven-family representation with cardinality vector
$\mathbf n_{25}$. The code $K_{4a}$ consists of $49$ ordered $4$-tuples of
labels whose corresponding Cartesian products are pairwise separated in at
least one coordinate. Their union is therefore an independent set. Since each
copy of the representation uses $25$ five-dimensional base blocks, the
resulting construction uses $4\cdot25=100$ base blocks and lies in
$C_7^{\strong 500}$. It gives the Lean-verified bound obtained by
BPZ~\cite{BuysPolakZuiddamGitHub2026}
\begin{align}
\Shannon(C_7)
\ge
3.258827985920007034526478965794221\ldots .
\end{align}

The improvement by BPZ~\cite{BuysPolakZuiddamGitHub2026} over Gao's binary
recursion comes from the substantially larger recursive design space made
available by the BPZ multi-gadget framework: the construction uses combining
rules of different arities, reuses previously constructed representations of
different dimensions, and incorporates both orientations of the
five-dimensional base representation.

\subsection{Improving the BPZ Recursion via Heterogeneity}
\label{subsec:BPZheterogeneousImprovement}

\noindent\textbf{Overview of our construction.}
The preceding subsection described the BPZ multi-gadget recursion for $C_7$,
whose final $125$-dimensional seven-family cardinality vector is
\begin{align}
\mathbf n_{25}
=
S_{3b}(\mathbf n_6,\mathbf n_{11},\mathbf n_8),
\end{align}
followed by an application of the terminal code $K_{4a}$ to four copies of a
representation with cardinality vector $\mathbf n_{25}$. As shown in Fig.~\ref{fig:hetBPZrecursion}, we retain this
top-level architecture: the ordered block split $25=6+11+8$, the combining
rule $S_{3b}$, and the terminal code $K_{4a}$ are unchanged. We instead
construct new cardinality vectors $\widetilde{\mathbf n}_6$,
$\widetilde{\mathbf n}_{11}$, and $\widetilde{\mathbf n}_8$ by applying
Theorem~\ref{thm:heteroGao} with explicit heterogeneous auxiliary choices.
The finite counts used in these choices, together with the exact final
arithmetic, are recorded in Appendix~\ref{app:C7-certificates}.

\begin{theorem}\label{thm:mainC7}
There exists an independent set in $C_7^{\strong 500}$ of cardinality
\begin{align}
M_\star
={}&
339646729181174569434085175737183332253173907546511006059711
\nonumber\\
&525903249092770468764004403539584283592629896215203947372766
\nonumber\\
&811966557767583863304708327733085591599903170761021314133740
\nonumber\\
&610665958348235292745609889506981587295019498664860797917849
\nonumber\\
&96498585401881281.
\end{align}
Consequently,
\begin{align}
\Shannon(C_7)
&\ge
M_\star^{1/500}
\nonumber\\
&=
3.2588326203532663091215390518104754376053875943219\ldots .
\end{align}
\end{theorem}

\begin{proof}
\vspace{5pt}
\noindent\textbf{\underline{Step $1$: Ten- and fifteen-dimensional ingredients.}}
We begin with the three oriented ten-dimensional Gao products
\begin{align}
\mathcal G_{10}
&=
{\color{red}\operatorname{Gao}}(\mathcal G_5,\mathcal G_5)[X_5;X_5],
\nonumber\\
\mathcal G_{10}^{A}
&=
{\color{red}\operatorname{Gao}}(\sigma\mathcal G_5,\mathcal G_5)[X_5;X_5],
\nonumber\\
\mathcal G_{10}^{D}
&=
{\color{red}\operatorname{Gao}}(\mathcal G_5,\sigma\mathcal G_5)[X_5;X_5],
\end{align}
with profiles
\begin{align}
\pi(\mathcal G_{10})
&=
(134753,5152,134689,105709,14490,14490),
\nonumber\\
\pi(\mathcal G_{10}^{A})
&=
(134753,5152,134689,105709,12236,16744),
\nonumber\\
\pi(\mathcal G_{10}^{D})
&=
(134753,5152,134689,105709,16744,12236).
\end{align}

\begin{figure}[t]
\centering
\includegraphics[width=\textwidth]{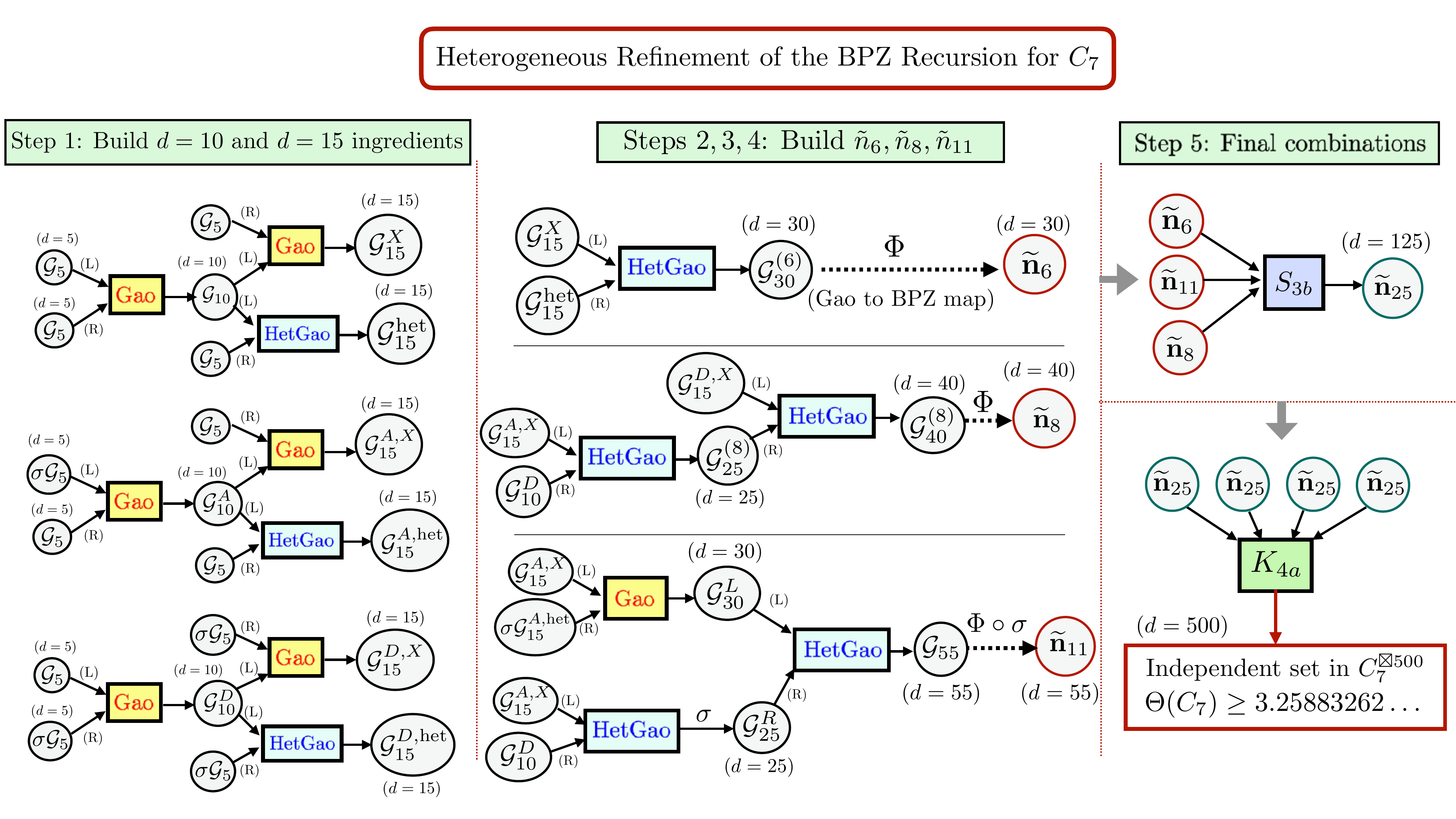}
\caption{Schematic of the heterogeneous refinement underlying
Theorem~\ref{thm:mainC7}. Step~1 constructs the oriented ten- and
fifteen-dimensional ingredients; Steps~2--4 produce the role-specific
cardinality vectors $\widetilde{\mathbf n}_6$,
$\widetilde{\mathbf n}_8$, and $\widetilde{\mathbf n}_{11}$; and
Step~5 applies the unchanged BPZ rule $S_{3b}$ and terminal code $K_{4a}$
to obtain an independent set in $C_7^{\strong 500}$. The labels $(L)$ and $(R)$ identify the first and second ordered
inputs, respectively, of each binary product, independently of
their vertical placement in the diagram.
The ordered product
arguments, auxiliary codebooks, and finite certificates are specified
in the proof and Appendix~\ref{app:C7-certificates}.}
\label{fig:hetBPZrecursion}
\vspace{-10pt}
\end{figure}

The superscripts $A$ and $D$ indicate the orientations relevant to the
corresponding BPZ coordinates. All three gadgets have the same physical
auxiliary set $X_{10}=X_5\times X_5$ and the same neutral part $X_{10}^0$.
Let $J^+$ be the independent set constructed in
Section~\ref{sec:heterogeneousGaoforC7}; recall that
\begin{align}
|J^+|
&=
134753,
&
\left|J^+\setminus N(X_{10}^0)\right|
&=
27488.
\end{align}

We next form three ordinary and three heterogeneous fifteen-dimensional
gadgets:
\begin{align}
\mathcal G_{15}^{X}
&=
{\color{red}\operatorname{Gao}}(\mathcal G_{10},\mathcal G_5)[X_{10};X_5],
\nonumber\\
\mathcal G_{15}^{\mathrm{het}}
&=
{\color{blue}\operatorname{HetGao}}(\mathcal G_{10},\mathcal G_5)
\bigl[(X_{10},J^+,J^+);X_5\bigr],
\nonumber\\
\mathcal G_{15}^{A,X}
&=
{\color{red}\operatorname{Gao}}(\mathcal G_{10}^{A},\mathcal G_5)[X_{10};X_5],
\nonumber\\
\mathcal G_{15}^{A,\mathrm{het}}
&=
{\color{blue}\operatorname{HetGao}}(\mathcal G_{10}^{A},\mathcal G_5)
\bigl[(X_{10},J^+,X_{10});X_5\bigr],
\nonumber\\
\mathcal G_{15}^{D,X}
&=
{\color{red}\operatorname{Gao}}(\mathcal G_{10}^{D},\sigma\mathcal G_5)[X_{10};X_5],
\nonumber\\
\mathcal G_{15}^{D,\mathrm{het}}
&=
{\color{blue}\operatorname{HetGao}}(\mathcal G_{10}^{D},\sigma\mathcal G_5)
\bigl[(X_{10},X_{10},J^+);X_5\bigr].
\end{align}
Their profiles are
\begin{align}
\pi(\mathcal G_{15}^{X})
&=
(49495055,2504616,49430863,35342398,6674251,7414214),
\nonumber\\
\pi(\mathcal G_{15}^{\mathrm{het}})
&=
(49495055,2504616,49433743,35275258,6703815,7454670),
\nonumber\\
\pi(\mathcal G_{15}^{A,X})
&=
(49495055,2504616,49430863,35342398,5948463,8140002),
\nonumber\\
\pi(\mathcal G_{15}^{A,\mathrm{het}})
&=
(49495055,2504616,49432527,35303606,5948463,8180458),
\nonumber\\
\pi(\mathcal G_{15}^{D,X})
&=
(49495055,2504616,49430863,35342398,8140002,5948463),
\nonumber\\
\pi(\mathcal G_{15}^{D,\mathrm{het}})
&=
(49495055,2504616,49432527,35303606,8180458,5948463).
\end{align}
The gadget $\mathcal G_{15}^{\mathrm{het}}$ is exactly the one constructed
in Section~\ref{sec:heterogeneousGaoforC7}; let $X_{15}$ denote its auxiliary
set. Let $X_{15}^{X}$ denote the common physical auxiliary set of the three
ordinary constructions, and let $X_{15}^{0,X}$ denote its neutral part. Let
$X_{15}^{A,\mathrm{het}}$ and $X_{15}^{D,\mathrm{het}}$ denote the auxiliary
sets of the two oriented heterogeneous constructions. These last two are the
same physical set, with their $\mathrm H$- and $\mathrm V$-parts
interchanged.

Let $J_{15}$ be the independent set defined in
Section~\ref{sec:heterogeneousGaoforC7}. Since it is an automorphic image of
a fifteen-dimensional main code,
\begin{align}
|J_{15}|=49495055.
\end{align}
The two finite counts needed below are
\begin{align}
\left|J_{15}\setminus N(X_{15}^{0,X})\right|
&=
12839823,
\nonumber\\
\left|X_{15}\setminus N(X_{15}^{0,X})\right|
&=
14045805.
\end{align}
Their verification is included in Appendix~\ref{app:C7-certificates}.

\vspace{5pt}
\noindent\textbf{\underline{Step $2$: Construction of $\widetilde{\mathbf n}_6$.}}
Define
\begin{align}
\mathcal G_{30}^{(6)}
=
{\color{blue}\operatorname{HetGao}}(\mathcal G_{15}^{X},\mathcal G_{15}^{\mathrm{het}})
\bigl[(X_{15},J_{15},J_{15});X_{15}\bigr].
\end{align}
The two input gadgets have the same product code, private pairs, and
transversals, so $X_{15}$ is admissible in the $J_0$ position relative to
$\mathcal G_{15}^{X}$, with
\begin{align}
(j_0,o_0,h_0,v_0)
=
(49433743,35275258,6703815,7454670).
\end{align}
The one-sided parameters are
\begin{align}
(j_{\mathrm H},q_{\mathrm H})
=
(j_{\mathrm V},q_{\mathrm V})
=
(49495055,12839823).
\end{align}
Theorem~\ref{thm:heteroGao} gives
\begin{align}
\pi(\mathcal G_{30}^{(6)})
=
(&2455719231434017,176870111100096,2444563032022369,\nonumber\\
&1426136268314719,509731462042710,508695301664940).
\end{align}
Applying the map $\Phi$ from \eqref{eq:Gao-to-BPZ} gives
\begin{align}
\widetilde{\mathbf n}_6
=
(&2278849120333921,1426136268314719,508695301664940,\nonumber\\
&509731462042710,176870111100096,\nonumber\\
&176870111100096,176870111100096).
\end{align}

\vspace{5pt}
\noindent\textbf{\underline{Step $3$: Construction of $\widetilde{\mathbf n}_8$.}}
First define
\begin{align}
\mathcal G_{25}^{(8)}
=
{\color{blue}\operatorname{HetGao}}(\mathcal G_{15}^{A,X},\mathcal G_{10}^{D})
\bigl[(X_{15}^{A,\mathrm{het}},J_{15},J_{15});X_{10}\bigr].
\end{align}
The set $X_{15}^{A,\mathrm{het}}$ is admissible in the $J_0$ position
because it is the auxiliary set of a gadget with the same product code,
private pairs, and transversals as $\mathcal G_{15}^{A,X}$, and its
decomposition is
\begin{align}
(j_0,o_0,h_0,v_0)
=
(49432527,35303606,5948463,8180458).
\end{align}
The resulting profile is
\begin{align}
\pi(\mathcal G_{25}^{(8)})
=
(&6682019915439,446844487240,6659829690543,\nonumber\\
&4104006957194,1077319494019,1478503239330).
\end{align}
Let $X_{25}^{(8)}$ denote its auxiliary set. Next define
\begin{align}
\mathcal G_{40}^{(8)}
=
{\color{blue}\operatorname{HetGao}}(\mathcal G_{15}^{D,X},\mathcal G_{25}^{(8)})
\bigl[(X_{15}^{D,\mathrm{het}},X_{15},J_{15});X_{25}^{(8)}\bigr].
\end{align}
Relative to $\mathcal G_{15}^{D,X}$, the decomposition of
$X_{15}^{D,\mathrm{het}}$ is
\begin{align}
(j_0,o_0,h_0,v_0)
=
(49432527,35303606,8180458,5948463).
\end{align}
The one-sided parameters are
\begin{align}
(j_{\mathrm H},q_{\mathrm H})
&=
(49433743,14045805),
\nonumber\\
(j_{\mathrm V},q_{\mathrm V})
&=
(49495055,12839823).
\end{align}
The resulting profile is
\begin{align}
\pi(\mathcal G_{40}^{(8)})
=
(&331761855244358723969,26071517201241409024,\nonumber\\
&329305968864222045505,179001784071649220449,\nonumber\\
&87767535795425989412,62536648997146835644).
\end{align}
Therefore,
\begin{align}
\widetilde{\mathbf n}_8
=
(&305690338043117314945,179001784071649220449,\nonumber\\
&62536648997146835644,87767535795425989412,\nonumber\\
&26071517201241409024,26071517201241409024,
26071517201241409024).
\end{align}

\vspace{5pt}
\noindent\textbf{\underline{Step $4$: Construction of $\widetilde{\mathbf n}_{11}$.}}
We first construct the thirty-dimensional left gadget using the ordinary Gao
product
\begin{align}
\mathcal G_{30}^{L}
=
{\color{red}\operatorname{Gao}}(\mathcal G_{15}^{A,X},
\sigma\mathcal G_{15}^{A,\mathrm{het}})
[X_{15}^{X};X_{15}^{D,\mathrm{het}}].
\end{align}
Its profile is
\begin{align}
\pi(\mathcal G_{30}^{L})
=
(&2455716185820961,176941111954464,2443492469880801,\nonumber\\
&1446768903083453,420235140891852,576488425905496).
\end{align}
Let $X_{30}^{0,L}$ denote its neutral part.

Using the same ordered pair of input gadgets, define the ``sibling gadget''
\begin{align}
\widehat{\mathcal G}_{30}
=
{\color{blue}\operatorname{HetGao}}(\mathcal G_{15}^{A,X},
\sigma\mathcal G_{15}^{A,\mathrm{het}})
\bigl[(X_{15}^{A,\mathrm{het}},J_{15},X_{15}^{X});
X_{15}^{D,\mathrm{het}}\bigr].
\end{align}
Its profile is
\begin{align}
\pi(\widehat{\mathcal G}_{30})
=
(&2455716185820961,176941111954464,2444076335041121,\nonumber\\
&1435184942161465,420235140891852,588656251987804).
\end{align}
Let $\widehat X_{30}$ denote its auxiliary set. Because
$\widehat{\mathcal G}_{30}$ and $\mathcal G_{30}^{L}$ have the same ordered
pair of input gadgets, their product code, propagated private pairs, and
transversals coincide. Hence $\widehat X_{30}$ is admissible in the $J_0$
position relative to $\mathcal G_{30}^{L}$, with
\begin{align}
(j_0,o_0,h_0,v_0)
={}&
(2444076335041121,1435184942161465,\nonumber\\
&420235140891852,588656251987804).
\end{align}

We next construct the twenty-five-dimensional right gadget:
\begin{align}
\widehat{\mathcal G}_{25}
=
{\color{blue}\operatorname{HetGao}}(\mathcal G_{15}^{A,X},\mathcal G_{10}^{D})
\bigl[(X_{15}^{A,\mathrm{het}},J_{15},X_{15});X_{10}\bigr].
\end{align}
Its profile is
\begin{align}
\pi(\widehat{\mathcal G}_{25})
=
(&6682019915439,446844487240,6659079476911,\nonumber\\
&4118763352946,1061812884635,1478503239330).
\end{align}
Reverse its orientation and write
\begin{align}
\mathcal G_{25}^{R}
&=
\sigma\widehat{\mathcal G}_{25},
\nonumber\\
\pi(\mathcal G_{25}^{R})
&=
(6682019915439,446844487240,6659079476911,\nonumber\\
&~~~~~4118763352946,1478503239330,1061812884635).
\end{align}
Let $X_{25}^{R}$ denote the auxiliary set of $\mathcal G_{25}^{R}$.

It remains to specify the independent set used on the two one-sided parts in
the final heterogeneous product. Let $\mathcal G_{30}^{++}$ be the ordinary
Gao product
\begin{align}
\mathcal G_{30}^{++}
=
{\color{red}\operatorname{Gao}}(\mathcal G_{15}^{\mathrm{het}},
\mathcal G_{15}^{\mathrm{het}})[X_{15};X_{15}],
\end{align}
and let $I_{30}^{++}$ denote its main independent set. Its size is
\begin{align}
|I_{30}^{++}|
&=
(49495055-2504616)^2
+2(2504616)(49433743)
\nonumber\\
&=
2455726444728097.
\end{align}
Recall the automorphism $T$ from
Section~\ref{sec:heterogeneousGaoforC7}. Apply $T$ independently to the six
five-coordinate blocks and set
\begin{align}
J_{30}^{++}
=
T^{\times 6}(I_{30}^{++}).
\end{align}
Then $J_{30}^{++}$ is independent and
\begin{align}
\left|J_{30}^{++}\setminus N(X_{30}^{0,L})\right|
=
841760069965664.
\end{align}
The seven-component verification of this count is given in
Appendix~\ref{app:C7-certificates}.

Finally, define
\begin{align}
\mathcal G_{55}
=
{\color{blue}\operatorname{HetGao}}(\mathcal G_{30}^{L},\mathcal G_{25}^{R})
\bigl[(\widehat X_{30},J_{30}^{++},J_{30}^{++});X_{25}^{R}\bigr].
\end{align}
Its profile is
\begin{align}
\pi(\mathcal G_{55})
=
(&16478688413981213775272008847,\nonumber\\
&1375259276200664518625890664,\nonumber\\
&16304893524159967980117037071,\nonumber\\
&8049523822718249447032963650,\nonumber\\
&3444579390015851502852512947,\nonumber\\
&4810790311425867030231560474).
\end{align}
Reversing the orientation and applying $\Phi$ gives
\begin{align}
\widetilde{\mathbf n}_{11}
=
(&15103429137780549256646118183,\nonumber\\
&8049523822718249447032963650,\nonumber\\
&3444579390015851502852512947,\nonumber\\
&4810790311425867030231560474,\nonumber\\
&1375259276200664518625890664,\nonumber\\
&1375259276200664518625890664,\nonumber\\
&1375259276200664518625890664).
\end{align}

\vspace{5pt}
\noindent\textbf{\underline{Step $5$: Final combinations.}}
We retain the ordered top-level BPZ rule and replace only its three inputs:
\begin{align}
\widetilde{\mathbf n}_{25}
=
S_{3b}(\widetilde{\mathbf n}_6,
\widetilde{\mathbf n}_{11},\widetilde{\mathbf n}_8).
\end{align}
By Proposition~\ref{prop:BPZgeneralizedCombiningRule},
$\widetilde{\mathbf n}_{25}$ is the cardinality vector of a seven-family
representation in $C_7^{\strong 125}$. Its exact entries are recorded in
Appendix~\ref{app:exact-final-arithmetic}. Finally, apply the unchanged BPZ
terminal code $K_{4a}$ to four copies of this representation. The resulting
independent set lies in $C_7^{\strong 500}$, and the exact terminal-code
cardinality formula gives $M_\star$ as stated in
Theorem~\ref{thm:mainC7}. This proves the theorem.
\end{proof}

The new bound strictly improves the BPZ \cite{BuysPolakZuiddamGitHub2026} value:
\begin{align}
\underbrace{3.258827985920007034526478965794221\ldots}_{\text{BPZ}~~ \cite{BuysPolakZuiddamGitHub2026}~\text{lower bound}}
<
\underbrace{3.258832620353266309121539051810475\ldots}_{\text{Theorem~}\ref{thm:mainC7}}.
\end{align}

The gain in the normalized lower bound is
$4.634433259274595060086\ldots \times 10^{-6}$.

\begin{remark}[\textbf{Intuition behind the construction}]
The improvement leaves the top-level BPZ split, the rule $S_{3b}$,
and the terminal code $K_{4a}$ unchanged. Rather than replacing
every ordinary gadget by a heterogeneous one, the construction
uses them for complementary purposes. For example, Step~2 uses
$\mathcal G_{15}^{X}$ as the left gadget but takes $J_0=X_{15}$
from $\mathcal G_{15}^{\mathrm{het}}$, which shares its transversals.
The ordinary gadget has a smaller auxiliary set but a larger
neutral part, trading some main-code words for more private pairs
in the next product. Theorem~\ref{thm:heteroGao} then allows the
heterogeneous choice of $J_0$ to modify the output auxiliary
structure without changing that product code or its private pairs.

Since $S_{3b}$ uses the constituent families differently in its
three ordered input positions, the branches can benefit from
different auxiliary distributions and orientations. The guiding
principle is therefore to choose intermediate profiles that work
better together in the final composition, rather than maximize
any single intermediate parameter.
\end{remark}

\section{Concluding Remarks and Open Problems}

This paper developed a heterogeneous refinement of recent recursive
constructions for lower-bounding zero-error Shannon capacity. Starting from
Gao's binary product construction, we showed that the auxiliary structure
propagated by a gadget need not be chosen uniformly across all parts of the
product. Allowing different independent sets to be used whenever the required
separation conditions permit leads to a strictly richer class of recursive
constructions. Even when such a heterogeneous choice does not improve the
size of the independent set immediately, it can produce a more useful
intermediate structure and thereby lead to a larger independent set at a
later stage of the recursion.
We then used this principle within the more general BPZ framework to construct
role-specific intermediate representations for different positions in the
recursion. Retaining the top-level BPZ split, combining rule, and terminal
code, we obtained an independent set in $C_7^{\strong 500}$ and the improved
lower bound
$\Shannon(C_7)
\ge
3.25883262\ldots$.
More broadly, these results suggest that recursive zero-error constructions
should not be optimized solely according to the size of the independent set
produced at each intermediate stage. Intermediate structures with the same
dimension and current main-code size can have different downstream value
because of the additional structure they carry. Recursive construction should
therefore be viewed as a structured, multi-objective optimization problem
rather than as a sequence of locally optimal choices.

Several directions appear particularly natural. First, the space of
heterogeneous codebooks, orientations, and recursive product trees has so far
been explored only sparsely. A systematic search should retain multiple
intermediate profiles, rather than only the construction with the largest
current code, and evaluate them according to their possible downstream roles.
More generally, one may jointly optimize the heterogeneous constructions,
BPZ combining rules, and terminal codes rather than keeping the top-level
architecture fixed. Second, a genuinely two-sided extension of
Theorem~\ref{thm:heteroGao} would allow both coordinates to vary (see Remark~\ref{remark:twosided}), but would
require additional compatibility conditions ensuring that every pair of
product blocks is separated in at least one coordinate. Third, it is natural
to investigate whether the same principles improve lower bounds for other
odd cycles and, more broadly, for other graphs whose Shannon capacities
remain unknown.

Finally, the recent progress on this problem reinforces the value of
zero-error Shannon capacity as a benchmark for AI-assisted mathematical
discovery. The target is precise, and every claimed improvement must come with a construction that can be checked independently. The real value, however, may lie less in the extra digits than in the ideas, search methods, and verification tools developed along the way.

\vspace{-5pt}
\section*{Code and Reproducibility}
\vspace{-5pt}
Code and machine-readable certificates for reproducing the constructions
and numerical bounds in this paper are available in the
accompanying repository~\cite{TandonC7Code2026}. These materials include
the finite neighborhood checks, recursive propagation routines, and
exact-arithmetic computations supporting the reported bounds. The
certificate data and final arithmetic for Theorem~\ref{thm:mainC7} are
recorded in Appendix~\ref{app:C7-certificates}.

\vspace{-5pt}
\section*{Disclosure of AI Use}\vspace{-5pt}

The author used large language models, including ChatGPT and Claude,
to explore mathematical ideas and candidate constructions, develop and
check proofs, implement computational searches and verification scripts,
and refine the manuscript. The mathematical claims are supported by the
proofs and finite computations presented in the paper. The author takes
full responsibility for the results and their presentation.

\bibliographystyle{IEEEtran}
\bibliography{references}


\vspace{-15pt}
\begin{appendices}

\vspace{-5pt}
\titleformat{\section}
  {\normalfont\Large\bfseries}
  {Appendix~\thesection:}
  {0.5em}
  {}

\section{BPZ Combining Rules and Terminal Codes}\label{app:BPZ-rules}

This appendix records the complete finite rule library used in the BPZ
multi-gadget framework.  The label set is
\begin{align}
\mathcal L=\{B,N,A,D,O,H,V\},
\end{align}
and tuple coordinates are ordered according to the input representations.
For a named combining rule $S_{m\alpha}$, we write
$T_{\lambda}^{(m\alpha)}\subseteq\mathcal L^m$ for its component associated
with the output label $\lambda\in\mathcal L$.  For a terminal code
$K_{m\alpha}$, the displayed set is the corresponding subset of
$\mathcal L^m$.

The lists below are transcribed from the files
\begin{center}
\texttt{ShannonBounds/Substitutions.lean},
\qquad
\texttt{ShannonBounds/TerminalCodes.lean},
\end{center}
at commit
\begin{center}
\texttt{aa21eeb12b75b0413d3fa9fb4208b5d0bf2c4d65}
\end{center}
of the BPZ Lean repository~\cite{BuysPolakZuiddamGitHub2026}.  Their total
word counts are, respectively,
\begin{align}
&20,17,58,45,48,46,53,47,47,45
\end{align}
for $S_{2a},S_{2b},S_{3a},\ldots,S_{3h}$, and
\begin{align}
19,49,57
\end{align}
for $K_{3a},K_{4a},K_{4b}$.  Admissibility of the combining rules and
pairwise separation of the terminal codes are verified in the cited Lean
source.

\subsection{Binary combining rules}

\subsubsection{The rule $S_{2a}$ (20 words)}\label{app:S2a}
\begin{longtable}{@{}>{\raggedright\arraybackslash}p{0.08\textwidth}>{\raggedright\arraybackslash}p{0.84\textwidth}@{}}
\toprule
Output & Ordered input words \\
\midrule
\endfirsthead
\toprule
Output & Ordered input words \\
\midrule
\endhead
$T_{B}^{(2a)}$ & $(B,B),\allowbreak\; (H,D),\allowbreak\; (V,A),\allowbreak\; (D,V),\allowbreak\; (A,H)$ \\
$T_{N}^{(2a)}$ & $(N,N),\allowbreak\; (A,A),\allowbreak\; (A,D),\allowbreak\; (D,A),\allowbreak\; (D,D)$ \\
$T_{A}^{(2a)}$ & $(A,N),\allowbreak\; (N,D)$ \\
$T_{D}^{(2a)}$ & $(D,N),\allowbreak\; (N,A)$ \\
$T_{O}^{(2a)}$ & $(O,N),\allowbreak\; (N,O)$ \\
$T_{H}^{(2a)}$ & $(H,N),\allowbreak\; (N,V)$ \\
$T_{V}^{(2a)}$ & $(V,N),\allowbreak\; (N,H)$ \\
\bottomrule
\end{longtable}
\addtocounter{table}{-1}

\subsubsection{The rule $S_{2b}$ (17 words)}\label{app:S2b}\vspace{-5pt}
\begin{longtable}{@{}>{\raggedright\arraybackslash}p{0.08\textwidth}>{\raggedright\arraybackslash}p{0.84\textwidth}@{}}
\toprule
Output & Ordered input words \\
\midrule
\endfirsthead
\toprule
Output & Ordered input words \\
\midrule
\endhead
$T_{B}^{(2b)}$ & $(A,V),\allowbreak\; (B,B),\allowbreak\; (D,H)$ \\
$T_{N}^{(2b)}$ & $(A,A),\allowbreak\; (N,N)$ \\
$T_{A}^{(2b)}$ & $(A,N),\allowbreak\; (B,D),\allowbreak\; (D,V),\allowbreak\; (V,H)$ \\
$T_{D}^{(2b)}$ & $(D,N),\allowbreak\; (N,A)$ \\
$T_{O}^{(2b)}$ & $\varnothing$ \\
$T_{H}^{(2b)}$ & $(H,B),\allowbreak\; (N,V)$ \\
$T_{V}^{(2b)}$ & $(H,A),\allowbreak\; (N,H),\allowbreak\; (V,D),\allowbreak\; (V,N)$ \\
\bottomrule
\end{longtable}
\addtocounter{table}{-1}

\subsection{Ternary combining rules}

\vspace{-5pt}

\subsubsection{The rule $S_{3a}$ (58 words)}\label{app:S3a}\vspace{-10pt}
\begin{longtable}{@{}>{\raggedright\arraybackslash}p{0.08\textwidth}>{\raggedright\arraybackslash}p{0.84\textwidth}@{}}
\toprule
Output & Ordered input words \\
\midrule
\endfirsthead
\toprule
Output & Ordered input words \\
\midrule
\endhead
$T_{B}^{(3a)}$ & $(A,H,N),\allowbreak\; (A,N,V),\allowbreak\; (B,A,H),\allowbreak\; (B,B,B),\allowbreak\; (B,D,V),\allowbreak\; (B,H,D),\allowbreak\; (B,V,A),\allowbreak\; (D,N,H),\allowbreak\; (D,V,N),\allowbreak\; (H,D,N),\allowbreak\; (H,N,A),\allowbreak\; (V,A,N),\allowbreak\; (V,N,D)$ \\
$T_{N}^{(3a)}$ & $(A,A,N),\allowbreak\; (A,D,N),\allowbreak\; (A,N,A),\allowbreak\; (A,N,D),\allowbreak\; (D,A,N),\allowbreak\; (D,D,N),\allowbreak\; (D,N,A),\allowbreak\; (D,N,D),\allowbreak\; (N,A,A),\allowbreak\; (N,A,D),\allowbreak\; (N,D,A),\allowbreak\; (N,D,D),\allowbreak\; (N,N,N)$ \\
$T_{A}^{(3a)}$ & $(A,A,A),\allowbreak\; (A,A,D),\allowbreak\; (A,D,A),\allowbreak\; (A,D,D),\allowbreak\; (A,N,N),\allowbreak\; (N,A,N),\allowbreak\; (N,N,A)$ \\
$T_{D}^{(3a)}$ & $(D,A,A),\allowbreak\; (D,A,D),\allowbreak\; (D,D,A),\allowbreak\; (D,D,D),\allowbreak\; (D,N,N),\allowbreak\; (N,D,N),\allowbreak\; (N,N,D)$ \\
$T_{O}^{(3a)}$ & $(O,A,A),\allowbreak\; (O,A,D),\allowbreak\; (O,D,A),\allowbreak\; (O,D,D)$ \\
$T_{H}^{(3a)}$ & $(H,A,A),\allowbreak\; (H,A,D),\allowbreak\; (H,D,A),\allowbreak\; (H,D,D),\allowbreak\; (H,N,N),\allowbreak\; (N,H,N),\allowbreak\; (N,N,H)$ \\
$T_{V}^{(3a)}$ & $(N,N,V),\allowbreak\; (N,V,N),\allowbreak\; (V,A,A),\allowbreak\; (V,A,D),\allowbreak\; (V,D,A),\allowbreak\; (V,D,D),\allowbreak\; (V,N,N)$ \\
\bottomrule
\end{longtable}
\addtocounter{table}{-1}

\vspace{-20pt}
\subsubsection{The rule $S_{3b}$ (45 words)}\label{app:S3b}\vspace{-10pt}
\begin{longtable}{@{}>{\raggedright\arraybackslash}p{0.08\textwidth}>{\raggedright\arraybackslash}p{0.84\textwidth}@{}}
\toprule
Output & Ordered input words \\
\midrule
\endfirsthead
\toprule
Output & Ordered input words \\
\midrule
\endhead
$T_{B}^{(3b)}$ & $(B,B,B),\allowbreak\; (B,H,A),\allowbreak\; (N,A,V),\allowbreak\; (N,D,H),\allowbreak\; (N,V,D),\allowbreak\; (A,B,H),\allowbreak\; (A,H,N),\allowbreak\; (A,H,D),\allowbreak\; (D,B,V),\allowbreak\; (D,V,N),\allowbreak\; (D,V,D),\allowbreak\; (H,N,D),\allowbreak\; (H,D,N),\allowbreak\; (H,D,D),\allowbreak\; (V,N,A),\allowbreak\; (V,A,B),\allowbreak\; (V,D,A)$ \\
$T_{N}^{(3b)}$ & $(B,V,D),\allowbreak\; (N,N,N),\allowbreak\; (N,V,A),\allowbreak\; (A,N,H),\allowbreak\; (A,D,H),\allowbreak\; (A,H,N),\allowbreak\; (D,N,V),\allowbreak\; (D,D,H),\allowbreak\; (D,H,N),\allowbreak\; (H,V,D)$ \\
$T_{A}^{(3b)}$ & $(N,B,A),\allowbreak\; (N,A,N),\allowbreak\; (A,B,B),\allowbreak\; (A,A,V),\allowbreak\; (A,H,A)$ \\
$T_{D}^{(3b)}$ & $(N,B,D),\allowbreak\; (N,D,N),\allowbreak\; (D,B,B),\allowbreak\; (D,A,V),\allowbreak\; (D,H,A)$ \\
$T_{O}^{(3b)}$ & $\varnothing$ \\
$T_{H}^{(3b)}$ & $(N,N,H),\allowbreak\; (N,H,N),\allowbreak\; (H,N,N),\allowbreak\; (H,V,V)$ \\
$T_{V}^{(3b)}$ & $(N,N,V),\allowbreak\; (N,V,N),\allowbreak\; (V,N,N),\allowbreak\; (V,V,V)$ \\
\bottomrule
\end{longtable}
\addtocounter{table}{-1}

\subsubsection{The rule $S_{3c}$ (48 words)}\label{app:S3c}\vspace{-10pt}
\begin{longtable}{@{}>{\raggedright\arraybackslash}p{0.08\textwidth}>{\raggedright\arraybackslash}p{0.84\textwidth}@{}}
\toprule
Output & Ordered input words \\
\midrule
\endfirsthead
\toprule
Output & Ordered input words \\
\midrule
\endhead
$T_{B}^{(3c)}$ & $(A,B,V),\allowbreak\; (A,H,H),\allowbreak\; (A,V,A),\allowbreak\; (B,B,B),\allowbreak\; (D,A,H),\allowbreak\; (D,H,A),\allowbreak\; (D,N,H),\allowbreak\; (H,A,A),\allowbreak\; (H,A,N),\allowbreak\; (H,H,H),\allowbreak\; (N,H,D),\allowbreak\; (N,V,A),\allowbreak\; (V,A,D),\allowbreak\; (V,D,B)$ \\
$T_{N}^{(3c)}$ & $(A,A,D),\allowbreak\; (A,D,A),\allowbreak\; (B,H,D),\allowbreak\; (D,A,B),\allowbreak\; (D,B,H),\allowbreak\; (H,H,A),\allowbreak\; (H,V,H),\allowbreak\; (N,N,N),\allowbreak\; (V,A,A),\allowbreak\; (V,D,H)$ \\
$T_{A}^{(3c)}$ & $(A,B,N),\allowbreak\; (A,N,D),\allowbreak\; (B,N,A),\allowbreak\; (B,V,V),\allowbreak\; (H,D,N),\allowbreak\; (N,A,B),\allowbreak\; (N,D,A),\allowbreak\; (V,N,V),\allowbreak\; (V,V,N)$ \\
$T_{D}^{(3c)}$ & $(D,N,N),\allowbreak\; (N,D,N),\allowbreak\; (N,N,D)$ \\
$T_{O}^{(3c)}$ & $\varnothing$ \\
$T_{H}^{(3c)}$ & $(B,H,N),\allowbreak\; (H,N,B),\allowbreak\; (N,B,H)$ \\
$T_{V}^{(3c)}$ & $(A,V,N),\allowbreak\; (D,H,N),\allowbreak\; (H,N,A),\allowbreak\; (N,A,H),\allowbreak\; (N,D,V),\allowbreak\; (N,N,V),\allowbreak\; (N,V,N),\allowbreak\; (V,N,D),\allowbreak\; (V,N,N)$ \\
\bottomrule
\end{longtable}
\addtocounter{table}{-1}

\subsubsection{The rule $S_{3d}$ (46 words)}\label{app:S3d}
\begin{longtable}{@{}>{\raggedright\arraybackslash}p{0.08\textwidth}>{\raggedright\arraybackslash}p{0.84\textwidth}@{}}
\toprule
Output & Ordered input words \\
\midrule
\endfirsthead
\toprule
Output & Ordered input words \\
\midrule
\endhead
$T_{B}^{(3d)}$ & $(A,B,V),\allowbreak\; (B,B,B),\allowbreak\; (B,H,D),\allowbreak\; (B,V,A),\allowbreak\; (D,A,H),\allowbreak\; (D,D,V),\allowbreak\; (D,N,H),\allowbreak\; (H,A,A),\allowbreak\; (H,A,N),\allowbreak\; (V,A,D),\allowbreak\; (V,D,N),\allowbreak\; (V,H,A),\allowbreak\; (V,H,H)$ \\
$T_{N}^{(3d)}$ & $(A,A,D),\allowbreak\; (B,H,D),\allowbreak\; (D,A,B),\allowbreak\; (D,B,H),\allowbreak\; (H,V,H),\allowbreak\; (N,N,N),\allowbreak\; (V,A,A),\allowbreak\; (V,D,H),\allowbreak\; (V,H,A)$ \\
$T_{A}^{(3d)}$ & $(A,B,N),\allowbreak\; (A,N,D),\allowbreak\; (B,D,A),\allowbreak\; (B,N,A),\allowbreak\; (B,V,V),\allowbreak\; (H,D,N),\allowbreak\; (N,A,B),\allowbreak\; (V,N,V),\allowbreak\; (V,V,N)$ \\
$T_{D}^{(3d)}$ & $(D,N,N),\allowbreak\; (N,D,N),\allowbreak\; (N,N,D)$ \\
$T_{O}^{(3d)}$ & $\varnothing$ \\
$T_{H}^{(3d)}$ & $(B,H,N),\allowbreak\; (H,N,B),\allowbreak\; (N,B,H)$ \\
$T_{V}^{(3d)}$ & $(A,V,N),\allowbreak\; (D,H,N),\allowbreak\; (H,N,A),\allowbreak\; (N,A,H),\allowbreak\; (N,D,V),\allowbreak\; (N,N,V),\allowbreak\; (N,V,N),\allowbreak\; (V,N,D),\allowbreak\; (V,N,N)$ \\
\bottomrule
\end{longtable}
\addtocounter{table}{-1}

\subsubsection{The rule $S_{3e}$ (53 words)}\label{app:S3e}
\begin{longtable}{@{}>{\raggedright\arraybackslash}p{0.08\textwidth}>{\raggedright\arraybackslash}p{0.84\textwidth}@{}}
\toprule
Output & Ordered input words \\
\midrule
\endfirsthead
\toprule
Output & Ordered input words \\
\midrule
\endhead
$T_{B}^{(3e)}$ & $(A,H,N),\allowbreak\; (B,A,H),\allowbreak\; (B,B,B),\allowbreak\; (B,D,V),\allowbreak\; (B,H,D),\allowbreak\; (B,V,A),\allowbreak\; (D,V,N),\allowbreak\; (H,D,N),\allowbreak\; (H,N,A),\allowbreak\; (V,A,N),\allowbreak\; (V,N,D)$ \\
$T_{N}^{(3e)}$ & $(A,A,N),\allowbreak\; (A,H,N),\allowbreak\; (D,A,N),\allowbreak\; (D,D,N),\allowbreak\; (H,N,A),\allowbreak\; (N,A,A),\allowbreak\; (N,A,D),\allowbreak\; (N,D,A),\allowbreak\; (N,D,D),\allowbreak\; (N,N,N)$ \\
$T_{A}^{(3e)}$ & $(A,A,A),\allowbreak\; (A,A,D),\allowbreak\; (A,D,A),\allowbreak\; (A,D,D),\allowbreak\; (A,N,N),\allowbreak\; (B,N,A),\allowbreak\; (H,N,D),\allowbreak\; (N,A,N)$ \\
$T_{D}^{(3e)}$ & $(A,N,D),\allowbreak\; (D,A,A),\allowbreak\; (D,A,D),\allowbreak\; (D,D,A),\allowbreak\; (D,D,D),\allowbreak\; (D,N,N),\allowbreak\; (N,D,N),\allowbreak\; (N,N,D)$ \\
$T_{O}^{(3e)}$ & $\varnothing$ \\
$T_{H}^{(3e)}$ & $(B,N,H),\allowbreak\; (H,A,A),\allowbreak\; (H,A,D),\allowbreak\; (H,D,D),\allowbreak\; (H,H,A),\allowbreak\; (H,N,N),\allowbreak\; (N,H,N)$ \\
$T_{V}^{(3e)}$ & $(A,N,V),\allowbreak\; (D,N,H),\allowbreak\; (N,N,V),\allowbreak\; (N,V,N),\allowbreak\; (V,A,A),\allowbreak\; (V,A,D),\allowbreak\; (V,D,A),\allowbreak\; (V,D,D),\allowbreak\; (V,N,N)$ \\
\bottomrule
\end{longtable}
\addtocounter{table}{-1}

\subsubsection{The rule $S_{3f}$ (47 words)}\label{app:S3f}
\begin{longtable}{@{}>{\raggedright\arraybackslash}p{0.08\textwidth}>{\raggedright\arraybackslash}p{0.84\textwidth}@{}}
\toprule
Output & Ordered input words \\
\midrule
\endfirsthead
\toprule
Output & Ordered input words \\
\midrule
\endhead
$T_{B}^{(3f)}$ & $(A,H,D),\allowbreak\; (A,H,N),\allowbreak\; (B,A,H),\allowbreak\; (B,B,B),\allowbreak\; (B,D,V),\allowbreak\; (B,V,A),\allowbreak\; (H,D,B),\allowbreak\; (H,H,H),\allowbreak\; (H,N,A),\allowbreak\; (N,H,D),\allowbreak\; (V,N,D)$ \\
$T_{N}^{(3f)}$ & $(A,H,B),\allowbreak\; (B,H,H),\allowbreak\; (D,H,A),\allowbreak\; (H,D,D),\allowbreak\; (H,N,A),\allowbreak\; (N,A,A),\allowbreak\; (N,A,H),\allowbreak\; (N,H,A),\allowbreak\; (N,N,N),\allowbreak\; (V,H,V)$ \\
$T_{A}^{(3f)}$ & $(A,A,A),\allowbreak\; (A,A,H),\allowbreak\; (A,B,N),\allowbreak\; (B,N,A),\allowbreak\; (H,H,N),\allowbreak\; (H,N,H),\allowbreak\; (N,A,N)$ \\
$T_{D}^{(3f)}$ & $(A,N,D),\allowbreak\; (D,A,A),\allowbreak\; (D,A,D),\allowbreak\; (D,B,N),\allowbreak\; (N,D,N),\allowbreak\; (N,N,D)$ \\
$T_{O}^{(3f)}$ & $\varnothing$ \\
$T_{H}^{(3f)}$ & $(B,N,H),\allowbreak\; (H,A,B),\allowbreak\; (H,A,V),\allowbreak\; (H,N,N),\allowbreak\; (N,H,N)$ \\
$T_{V}^{(3f)}$ & $(A,N,V),\allowbreak\; (D,N,H),\allowbreak\; (D,V,N),\allowbreak\; (N,N,V),\allowbreak\; (N,V,N),\allowbreak\; (V,A,B),\allowbreak\; (V,A,V),\allowbreak\; (V,N,N)$ \\
\bottomrule
\end{longtable}
\addtocounter{table}{-1}

\subsubsection{The rule $S_{3g}$ (47 words)}\label{app:S3g}
\begin{longtable}{@{}>{\raggedright\arraybackslash}p{0.08\textwidth}>{\raggedright\arraybackslash}p{0.84\textwidth}@{}}
\toprule
Output & Ordered input words \\
\midrule
\endfirsthead
\toprule
Output & Ordered input words \\
\midrule
\endhead
$T_{B}^{(3g)}$ & $(A,H,N),\allowbreak\; (B,A,H),\allowbreak\; (B,B,B),\allowbreak\; (B,D,V),\allowbreak\; (B,H,D),\allowbreak\; (B,V,A),\allowbreak\; (H,D,N),\allowbreak\; (H,N,A),\allowbreak\; (V,H,A),\allowbreak\; (V,H,H),\allowbreak\; (V,N,D)$ \\
$T_{N}^{(3g)}$ & $(A,H,B),\allowbreak\; (B,H,H),\allowbreak\; (D,H,A),\allowbreak\; (H,D,D),\allowbreak\; (H,N,A),\allowbreak\; (N,A,A),\allowbreak\; (N,A,H),\allowbreak\; (N,H,A),\allowbreak\; (N,N,N),\allowbreak\; (V,H,V)$ \\
$T_{A}^{(3g)}$ & $(A,A,A),\allowbreak\; (A,A,H),\allowbreak\; (A,B,N),\allowbreak\; (B,N,A),\allowbreak\; (H,H,N),\allowbreak\; (H,N,H),\allowbreak\; (N,A,N)$ \\
$T_{D}^{(3g)}$ & $(A,N,D),\allowbreak\; (D,A,A),\allowbreak\; (D,A,H),\allowbreak\; (D,B,N),\allowbreak\; (N,D,N),\allowbreak\; (N,N,D)$ \\
$T_{O}^{(3g)}$ & $\varnothing$ \\
$T_{H}^{(3g)}$ & $(B,N,H),\allowbreak\; (H,A,B),\allowbreak\; (H,A,V),\allowbreak\; (H,N,N),\allowbreak\; (N,H,N)$ \\
$T_{V}^{(3g)}$ & $(A,N,V),\allowbreak\; (D,N,H),\allowbreak\; (D,V,N),\allowbreak\; (N,N,V),\allowbreak\; (N,V,N),\allowbreak\; (V,A,B),\allowbreak\; (V,A,V),\allowbreak\; (V,N,N)$ \\
\bottomrule
\end{longtable}
\addtocounter{table}{-1}

\subsubsection{The rule $S_{3h}$ (45 words)}\label{app:S3h}
\begin{longtable}{@{}>{\raggedright\arraybackslash}p{0.08\textwidth}>{\raggedright\arraybackslash}p{0.84\textwidth}@{}}
\toprule
Output & Ordered input words \\
\midrule
\endfirsthead
\toprule
Output & Ordered input words \\
\midrule
\endhead
$T_{B}^{(3h)}$ & $(B,B,B),\allowbreak\; (B,H,A),\allowbreak\; (N,A,V),\allowbreak\; (N,D,H),\allowbreak\; (N,V,D),\allowbreak\; (A,B,H),\allowbreak\; (A,H,N),\allowbreak\; (A,H,D),\allowbreak\; (D,B,V),\allowbreak\; (D,V,N),\allowbreak\; (D,V,D),\allowbreak\; (H,N,D),\allowbreak\; (H,D,N),\allowbreak\; (H,D,D),\allowbreak\; (V,N,A),\allowbreak\; (V,A,B),\allowbreak\; (V,D,A)$ \\
$T_{N}^{(3h)}$ & $(B,V,D),\allowbreak\; (N,N,N),\allowbreak\; (N,V,A),\allowbreak\; (A,N,H),\allowbreak\; (A,D,H),\allowbreak\; (A,H,N),\allowbreak\; (D,N,V),\allowbreak\; (D,D,H),\allowbreak\; (D,H,N),\allowbreak\; (H,V,D)$ \\
$T_{A}^{(3h)}$ & $(N,B,D),\allowbreak\; (N,D,N),\allowbreak\; (D,B,B),\allowbreak\; (D,A,V),\allowbreak\; (D,H,A)$ \\
$T_{D}^{(3h)}$ & $(N,B,A),\allowbreak\; (N,A,N),\allowbreak\; (A,B,B),\allowbreak\; (A,A,V),\allowbreak\; (A,H,A)$ \\
$T_{O}^{(3h)}$ & $\varnothing$ \\
$T_{H}^{(3h)}$ & $(N,N,V),\allowbreak\; (N,V,N),\allowbreak\; (V,N,N),\allowbreak\; (V,V,V)$ \\
$T_{V}^{(3h)}$ & $(N,N,H),\allowbreak\; (N,H,N),\allowbreak\; (H,N,N),\allowbreak\; (H,V,V)$ \\
\bottomrule
\end{longtable}
\addtocounter{table}{-1}

\subsection{Terminal codes}

\subsubsection{The terminal code $K_{3a}$ (19 words)}\label{app:K3a}
\begin{center}
\scriptsize
\renewcommand{\arraystretch}{1.16}
\begin{tabular}{@{}ccc@{}}
$(A,B,V)$ & $(A,H,N)$ & $(A,V,H)$ \\
$(B,A,B)$ & $(B,D,D)$ & $(B,H,A)$ \\
$(B,N,B)$ & $(D,A,H)$ & $(D,D,N)$ \\
$(D,N,H)$ & $(H,B,A)$ & $(H,H,H)$ \\
$(H,V,N)$ & $(N,B,V)$ & $(N,H,N)$ \\
$(N,V,H)$ & $(V,B,D)$ & $(V,B,N)$ \\
$(V,V,A)$ &  &  \\
\end{tabular}
\end{center}

\subsubsection{The terminal code $K_{4a}$ (49 words)}\label{app:K4a}
\begin{center}
\scriptsize
\renewcommand{\arraystretch}{1.16}
\begin{tabular}{@{}ccc@{}}
$(B,B,B,B)$ & $(B,N,V,B)$ & $(B,A,H,B)$ \\
$(B,D,V,B)$ & $(B,H,N,B)$ & $(B,H,D,B)$ \\
$(B,V,A,B)$ & $(N,B,B,V)$ & $(N,N,V,H)$ \\
$(N,A,V,H)$ & $(N,D,H,H)$ & $(N,H,N,H)$ \\
$(N,H,A,H)$ & $(N,V,D,V)$ & $(A,B,B,V)$ \\
$(A,N,H,H)$ & $(A,A,V,H)$ & $(A,D,H,H)$ \\
$(A,H,A,H)$ & $(A,V,N,H)$ & $(A,V,D,H)$ \\
$(D,B,B,H)$ & $(D,N,V,V)$ & $(D,A,V,V)$ \\
$(D,D,H,V)$ & $(D,H,N,V)$ & $(D,H,A,V)$ \\
$(D,V,D,V)$ & $(H,B,B,A)$ & $(H,N,V,D)$ \\
$(H,A,H,D)$ & $(H,A,V,N)$ & $(H,D,H,N)$ \\
$(H,D,V,D)$ & $(H,H,N,N)$ & $(H,H,N,D)$ \\
$(H,H,A,N)$ & $(H,H,D,D)$ & $(H,V,A,D)$ \\
$(H,V,D,N)$ & $(V,B,B,N)$ & $(V,B,B,D)$ \\
$(V,N,V,N)$ & $(V,N,V,A)$ & $(V,A,H,A)$ \\
$(V,D,V,A)$ & $(V,H,D,A)$ & $(V,V,N,A)$ \\
$(V,V,A,A)$ &  &  \\
\end{tabular}
\end{center}

\subsubsection{The terminal code $K_{4b}$ (57 words)}\label{app:K4b}
\begin{center}
\scriptsize
\renewcommand{\arraystretch}{1.16}
\begin{tabular}{@{}ccc@{}}
$(A,A,A,V)$ & $(A,A,N,V)$ & $(A,B,H,A)$ \\
$(A,H,A,A)$ & $(A,H,A,N)$ & $(A,H,N,A)$ \\
$(A,H,N,N)$ & $(A,N,A,V)$ & $(A,N,D,H)$ \\
$(A,N,H,N)$ & $(A,N,N,H)$ & $(A,V,D,B)$ \\
$(B,A,D,H)$ & $(B,A,V,N)$ & $(B,B,B,B)$ \\
$(B,B,V,D)$ & $(B,D,B,H)$ & $(B,D,H,N)$ \\
$(B,H,H,H)$ & $(B,V,A,D)$ & $(B,V,N,D)$ \\
$(D,A,V,A)$ & $(D,H,H,A)$ & $(D,N,B,V)$ \\
$(D,N,V,A)$ & $(D,N,V,N)$ & $(D,V,B,A)$ \\
$(D,V,B,N)$ & $(H,A,B,H)$ & $(H,B,D,B)$ \\
$(H,B,V,H)$ & $(H,D,A,B)$ & $(H,D,N,B)$ \\
$(H,N,N,D)$ & $(H,V,H,D)$ & $(H,V,H,N)$ \\
$(N,A,A,V)$ & $(N,A,N,V)$ & $(N,B,H,A)$ \\
$(N,H,A,N)$ & $(N,H,B,A)$ & $(N,N,A,V)$ \\
$(N,N,D,H)$ & $(N,N,N,V)$ & $(N,N,V,N)$ \\
$(N,V,D,N)$ & $(N,V,H,A)$ & $(N,V,N,N)$ \\
$(V,A,A,B)$ & $(V,A,N,B)$ & $(V,D,D,H)$ \\
$(V,H,A,H)$ & $(V,H,N,H)$ & $(V,N,A,B)$ \\
$(V,N,N,A)$ & $(V,N,N,N)$ & $(V,V,H,V)$ \\
\end{tabular}
\end{center}

\section{Finite Certificates and Exact Arithmetic for Improved $C_7$ Bounds}
\label{app:C7-certificates}

This appendix records the finite certificate data supporting both improved
lower bounds, together with the final exact arithmetic for Theorem~\ref{thm:mainC7}.
Throughout, $N(S)$ denotes the closed confusability neighborhood of a set $S$.
\subsection{The eight exchanges defining $J^+$}
\label{app:Jplus}

For a word $u=(u_0,u_1,u_2,u_3,u_4)\in\mathbb Z_7^5$, use the base-$7$
encoding
\begin{align}
\operatorname{enc}(u)=\sum_{k=0}^{4}u_k7^k.
\end{align}
A word of $\mathbb Z_7^{10}$ is written as an ordered pair of such
five-dimensional encodings. Let
\begin{align}\label{eq:appendix-T}
T(w_0,w_1,w_2,w_3,w_4)
=
(2-w_1,w_3,w_0,2-w_2,w_4)
\pmod 7
\end{align}
and set
\begin{align}
J=(T\times T)(I_{10}).
\end{align}
Direct enumeration gives
\begin{align}
\left|J\setminus N(X_{10}^0)\right|=27{,}480.
\end{align}
The following eight simultaneous exchanges define $J^+$. In each line,
the word on the left is removed and the word on the right is inserted:
\begin{align}
(3660,6494)&\longrightarrow(6110,4092),
&
(16388,6494)&\longrightarrow(16046,4092),
\nonumber\\
(6494,3660)&\longrightarrow(6494,6110),
&
(15643,3660)&\longrightarrow(15643,6110),
\nonumber\\
(3660,15643)&\longrightarrow(6110,15685),
&
(16388,15643)&\longrightarrow(16046,15685),
\nonumber\\
(6494,16388)&\longrightarrow(6494,16046),
&
(15643,16388)&\longrightarrow(15643,16046).
\end{align}
Each inserted word has exactly one confusable neighbor in $J$, namely the
word removed on the same line, and the eight inserted words are pairwise
nonconfusable. The exchanges therefore preserve independence and
cardinality. Moreover, in each exchange the removed word belongs to
$N(X_{10}^0)$, whereas the inserted word does not. Consequently,
\begin{align}\label{eq:appendix-q10}
|J^+|=134{,}753,
\qquad
q_{10}:=\left|J^+\setminus N(X_{10}^0)\right|=27{,}488.
\end{align}

\subsection{Certificate ledger}
\label{app:certificate-ledger}

Using the notation of Section~\ref{subsec:BPZheterogeneousImprovement}, let
$X_{15}^{0,X}$ denote the neutral part of the common auxiliary set
$X_{15}^{X}$ of the three ordinary fifteen-dimensional constructions, and
let $X_{30}^{0,L}$ denote the neutral part of the thirty-dimensional left
reference gadget $\mathcal G_{30}^{L}$.

Let
\begin{align}
J_{15}&=(T\times T\times T)(I_{15}),
\end{align}
where $I_{15}$ is the common main code of $\mathcal G_{15}^{X}$ and
$\mathcal G_{15}^{\mathrm{het}}$. The two fifteen-dimensional auxiliary
sets used in the certificates are
\begin{align}
X_{15}^{X}
&=
(X_{10}\times X_5^0)
\mathbin{\dot\cup}
(X_{10}\times X_5^{\mathrm H})
\mathbin{\dot\cup}
(X_{10}\times X_5^{\mathrm V}),
\label{eq:appendix-X15X}\\
X_{15}
&=
(X_{10}\times X_5^0)
\mathbin{\dot\cup}
(J^+\times X_5^{\mathrm H})
\mathbin{\dot\cup}
(J^+\times X_5^{\mathrm V}).
\label{eq:appendix-X15het}
\end{align}
Thus $X_{15}$ is independent because it is the auxiliary set of
$\mathcal G_{15}^{\mathrm{het}}$. The sets $J_{15}$ and $J_{30}^{++}$ are
independent because they are automorphic images of main independent sets.
Specifically, let $I_{30}^{++}$ denote the main independent set of
\begin{align}
\mathcal G_{30}^{++}
=
\operatorname{Gao}(\mathcal G_{15}^{\mathrm{het}},
\mathcal G_{15}^{\mathrm{het}})[X_{15};X_{15}].
\end{align}
Then
\begin{align}\label{eq:appendix-J30pp}
J_{30}^{++}=T^{\times6}(I_{30}^{++}),
\end{align}
and
\begin{align}\label{eq:appendix-size-J30pp}
|J_{30}^{++}|=|I_{30}^{++}|=2{,}455{,}726{,}444{,}728{,}097.
\end{align}
The final construction depends on the following four nontrivial neighborhood
counts.

\begin{table}[htbp]
\centering
\small
\renewcommand{\arraystretch}{1.18}
\begin{tabularx}{0.96\textwidth}{@{}c >{\raggedright\arraybackslash}X r@{}}
\toprule
\textbf{Certificate} & \textbf{Definition} & \textbf{Value} \\
\midrule
$C_1$ & $q_{10}=|J^+\setminus N(X_{10}^0)|$ & $27{,}488$ \\
$C_2$ & $|J_{15}\setminus N(X_{15}^{0,X})|$ & $12{,}839{,}823$ \\
$C_3$ & $|X_{15}\setminus N(X_{15}^{0,X})|$ & $14{,}045{,}805$ \\
$C_4$ & $q_{30}^{++}:=|J_{30}^{++}\setminus N(X_{30}^{0,L})|$ & $841{,}760{,}069{,}965{,}664$ \\
\bottomrule
\end{tabularx}
\caption{The four nontrivial finite certificates in the final dependency
chain. Each value is recorded together with the neutral reference against
which it is evaluated.}
\label{tab:appendix-certificates}
\end{table}
\paragraph{Additional certificate for the ``Warmup'' in Section~\ref{sec:heterogeneousGaoforC7}.}
The construction in Section~\ref{sec:heterogeneousGaoforC7} uses the additional count
\[
q_{15}
=
\bigl|J_{15}\setminus N(X_{15}^{0})\bigr|
=
12{,}872{,}271,
\]
where $X_{15}^{0}$ is the neutral part of the auxiliary set
of $\mathcal G_{15}^{\mathrm{het}}$.
This reference set differs from $X_{15}^{0,X}$ used in
certificate $C_2$.

\subsection{Component verification of $q_{30}^{++}$}
\label{app:q30-component-verification}

Write the main code and auxiliary decomposition of
$\mathcal G_{15}^{\mathrm{het}}$ as
$I_{15}=B_{15}\mathbin{\dot\cup}R_{15}$ and
$X_{15}=X_{15}^0\mathbin{\dot\cup}X_{15}^{\mathrm H}
\mathbin{\dot\cup}X_{15}^{\mathrm V}$, with complementary transversals
$P_{15}^{\mathrm H}$ and $P_{15}^{\mathrm V}$. The main independent set
$I_{30}^{++}$ is the disjoint union of the seven Gao-product components
shown in Table~\ref{tab:appendix-q30-components}. The second numerical
column gives the contribution of the corresponding component, after applying
$T^{\times6}$, to $J_{30}^{++}\setminus N(X_{30}^{0,L})$. The first total agrees with \eqref{eq:appendix-size-J30pp}, while the second
is certificate $C_4$ in Table~\ref{tab:appendix-certificates}.

\begin{table}[htbp]
\centering
\small
\renewcommand{\arraystretch}{1.18}
\begin{tabular}{@{}lrr@{}}
\toprule
\textbf{Gao component} & \textbf{Component size} & \textbf{Contribution to $q_{30}^{++}$} \\
\midrule
$B_{15}\times B_{15}$ & $2{,}208{,}101{,}357{,}412{,}721$ & $827{,}599{,}599{,}610{,}592$ \\
$R_{15}\times X_{15}^0$ & $88{,}350{,}975{,}590{,}928$ & $6{,}686{,}131{,}598{,}464$ \\
$P_{15}^{\mathrm H}\times X_{15}^{\mathrm H}$ & $16{,}790{,}482{,}310{,}040$ & $194{,}829{,}147{,}616$ \\
$P_{15}^{\mathrm V}\times X_{15}^{\mathrm V}$ & $18{,}671{,}085{,}756{,}720$ & $194{,}829{,}147{,}616$ \\
$X_{15}^0\times R_{15}$ & $88{,}350{,}975{,}590{,}928$ & $6{,}694{,}533{,}241{,}984$ \\
$X_{15}^{\mathrm H}\times P_{15}^{\mathrm V}$ & $16{,}790{,}482{,}310{,}040$ & $195{,}073{,}609{,}696$ \\
$X_{15}^{\mathrm V}\times P_{15}^{\mathrm H}$ & $18{,}671{,}085{,}756{,}720$ & $195{,}073{,}609{,}696$ \\
\midrule
\textbf{Total} & $2{,}455{,}726{,}444{,}728{,}097$ & $841{,}760{,}069{,}965{,}664$ \\
\bottomrule
\end{tabular}
\caption{Seven-component verification of the dimension-$30$ certificate
$q_{30}^{++}$.}
\label{tab:appendix-q30-components}
\end{table}

To evaluate these contributions, we use the product structure of the
neutral reference set. Let $U_0$ and $V_0$ be the neutral auxiliary
parts of the two ordered inputs defining $\mathcal G_{30}^{\mathrm L}$,
and let $U_1$ and $V_1$ be their respective nonneutral auxiliary parts.
By Lemma~\ref{lemma:GaoProduct},
\[
X_{30}^{0,\mathrm L}
=
(U_0\times V_0)
\mathbin{\dot\cup}
(U_1\times V_1).
\]
Since closed neighborhoods distribute over unions and satisfy
$N(A\times B)=N(A)\times N(B)$, a point $(u,v)$ belongs to
$N(X_{30}^{0,\mathrm L})$ precisely when
\[
\bigl(u\in N(U_0)\ \text{and}\ v\in N(V_0)\bigr)
\quad\text{or}\quad
\bigl(u\in N(U_1)\ \text{and}\ v\in N(V_1)\bigr).
\]
For each transformed factor appearing in Table~\ref{tab:appendix-q30-components}, the verification
therefore records the joint frequencies of the two corresponding
neighborhood-membership indicators. The contribution of each product
component is then obtained by exact products and sums of these
frequencies. Joint frequencies are used so that overlaps between the
two neighborhood conditions are counted correctly.

\subsection{Exact final arithmetic}
\label{app:exact-final-arithmetic}

The three seven-family cardinality vectors entering the final BPZ rule,
recorded in the order $(B,N,A,D,O,H,V)$, are
\begin{align}
\widetilde{\mathbf n}_6
={}&(2278849120333921,1426136268314719,508695301664940,\notag\\
&509731462042710,176870111100096,176870111100096,176870111100096),
\label{eq:appendix-n6}\\
\widetilde{\mathbf n}_{11}
={}&(15103429137780549256646118183,\notag\\
&8049523822718249447032963650,\notag\\
&3444579390015851502852512947,\notag\\
&4810790311425867030231560474,\notag\\
&1375259276200664518625890664,\notag\\
&1375259276200664518625890664,\notag\\
&1375259276200664518625890664),
\label{eq:appendix-n11}\\
\widetilde{\mathbf n}_8
={}&(305690338043117314945,179001784071649220449,\notag\\
&62536648997146835644,87767535795425989412,\notag\\
&26071517201241409024,26071517201241409024,26071517201241409024).
\label{eq:appendix-n8}
\end{align}
Applying the ternary rule $S_{3b}$ in the ordered split $25=6+11+8$ gives
\begin{align}
\widetilde{\mathbf n}_{25}
={}&(12651555102711866584006925012965893561076969596025225905818339799,\notag\\
&3066132745665704601020173586804414424611769880543231510544034510,\notag\\
&4664412231532807578550047842170656906931895714626194718133859805,\notag\\
&5561607824514622663538712910059154540700557990699908637174352708,\notag\\
&0,\notag\\
&911561283323995002333709184151818987697897558779375600768623640,\notag\\
&911561283323995002333709184151818987697897558779375600768623640).
\label{eq:appendix-n25}
\end{align}
The zero $O$-coordinate is consistent with
$T_O^{(3b)}=\varnothing$. Applying the terminal code $K_{4a}$ to four
copies of \eqref{eq:appendix-n25} produces an independent set in
$C_7^{\strong 500}$ of cardinality
\begin{align}\label{eq:appendix-M}
\begin{split}
M_\star={}&3396467291811745694340851757371833322531739075465110\\
&0605971152590324909277046876400440353958428359262989\\
&6215203947372766811966557767583863304708327733085591\\
&5999031707610213141337406106659583482352927456098895\\
&0698158729501949866486079791784996498585401881281.
\end{split}
\end{align}
Consequently,
\begingroup
\small
\begin{align}
\Shannon(C_7)
\ge
M_\star^{1/500}
=
3.2588326203532663091215390518104754376053875943219555178734747247104368\ldots .
\end{align}
\endgroup

\subsection{Input parameters for the heterogeneous products}
\label{app:heterogeneous-inputs}

Table~\ref{tab:appendix-heterogeneous-inputs} collects the input parameters
for the nine applications of Theorem~\ref{thm:heteroGao} in
Section~\ref{subsec:BPZheterogeneousImprovement}.
For each ordered pair $(\mathcal G_L,\mathcal G_R)$, the parameters
$(j_0,o_0,h_0,v_0)$ describe the decomposition of $J_0$ relative to
the transversals of $\mathcal G_L$, while
$
q_{\mathrm H}=|J_{\mathrm H}\setminus N(X_L^0)|,
q_{\mathrm V}=|J_{\mathrm V}\setminus N(X_L^0)|$
are evaluated against the neutral auxiliary part of that same left gadget.

Some of these counts follow directly from independence.
For any independent set $X$ and subset $S\subseteq X$,
$X\cap N(S)=S$,
because the neighborhood is closed and no point of $X\setminus S$
is confusable with a point of $S$.
Consequently, when a one-sided codebook equals the left auxiliary
set $X_L$, its corresponding $q$-value is $s_L-o_L$.
The two such counts used below are
\begin{align*}
|X_{10}\setminus N(X_{10}^0)|
&=134689-105709=28980,\\
|X_{15}^{X}\setminus N(X_{15}^{0,X})|
&=49430863-35342398=14088465.
\end{align*}
The remaining $q$-values in
Table~\ref{tab:appendix-heterogeneous-inputs} are the certificates
$C_1,C_2,C_3,C_4$ in Table~\ref{tab:appendix-certificates}.

\newpage 
\begin{table}[htbp]
\centering
\footnotesize
\setlength{\tabcolsep}{3.5pt}
\renewcommand{\arraystretch}{1.2}

\textbf{(a) Ordered gadget inputs and $J_0$ parameters.}

\smallskip
\begin{tabularx}{\textwidth}
{@{}lll>{\raggedright\arraybackslash}X@{}}
\toprule
\textbf{Output gadget}
& $(\mathcal G_L,\mathcal G_R)$
& $J_0$
& $(j_0,o_0,h_0,v_0)$ \\
\midrule

$\mathcal G_{15}^{\mathrm{het}}$
& $(\mathcal G_{10},\mathcal G_5)$
& $X_{10}$
& $(134689,105709,14490,14490)$ \\

$\mathcal G_{15}^{A,\mathrm{het}}$
& $(\mathcal G_{10}^{A},\mathcal G_5)$
& $X_{10}$
& $(134689,105709,12236,16744)$ \\

$\mathcal G_{15}^{D,\mathrm{het}}$
& $(\mathcal G_{10}^{D},\sigma\mathcal G_5)$
& $X_{10}$
& $(134689,105709,16744,12236)$ \\

$\mathcal G_{30}^{(6)}$
& $(\mathcal G_{15}^{X},\mathcal G_{15}^{\mathrm{het}})$
& $X_{15}$
& $(49433743,35275258,6703815,7454670)$ \\

$\mathcal G_{25}^{(8)}$
& $(\mathcal G_{15}^{A,X},\mathcal G_{10}^{D})$
& $X_{15}^{A,\mathrm{het}}$
& $(49432527,35303606,5948463,8180458)$ \\

$\mathcal G_{40}^{(8)}$
& $(\mathcal G_{15}^{D,X},\mathcal G_{25}^{(8)})$
& $X_{15}^{D,\mathrm{het}}$
& $(49432527,35303606,8180458,5948463)$ \\

$\widehat{\mathcal G}_{30}$
& $(\mathcal G_{15}^{A,X},\sigma\mathcal G_{15}^{A,\mathrm{het}})$
& $X_{15}^{A,\mathrm{het}}$
& $(49432527,35303606,5948463,8180458)$ \\

$\widehat{\mathcal G}_{25}$
& $(\mathcal G_{15}^{A,X},\mathcal G_{10}^{D})$
& $X_{15}^{A,\mathrm{het}}$
& $(49432527,35303606,5948463,8180458)$ \\

$\mathcal G_{55}$
& $(\mathcal G_{30}^{L},\mathcal G_{25}^{R})$
& $\widehat X_{30}$
& $\begin{aligned}
   &(2444076335041121,1435184942161465,\\
   &\phantom{(}420235140891852,588656251987804)
   \end{aligned}$ \\

\bottomrule
\end{tabularx}

\medskip
\textbf{(b) One-sided codebooks and neighborhood counts.}

\smallskip
\begin{tabularx}{\textwidth}
{@{}ll>{\raggedleft\arraybackslash}X
l>{\raggedleft\arraybackslash}X@{}}
\toprule
\textbf{Output gadget}
& $J_{\mathrm H}$
& $(j_{\mathrm H},q_{\mathrm H})$
& $J_{\mathrm V}$
& $(j_{\mathrm V},q_{\mathrm V})$ \\
\midrule

$\mathcal G_{15}^{\mathrm{het}}$
& $J^+$
& $(134753,27488)$
& $J^+$
& $(134753,27488)$ \\

$\mathcal G_{15}^{A,\mathrm{het}}$
& $J^+$
& $(134753,27488)$
& $X_{10}$
& $(134689,28980)$ \\

$\mathcal G_{15}^{D,\mathrm{het}}$
& $X_{10}$
& $(134689,28980)$
& $J^+$
& $(134753,27488)$ \\

$\mathcal G_{30}^{(6)}$
& $J_{15}$
& $(49495055,12839823)$
& $J_{15}$
& $(49495055,12839823)$ \\

$\mathcal G_{25}^{(8)}$
& $J_{15}$
& $(49495055,12839823)$
& $J_{15}$
& $(49495055,12839823)$ \\

$\mathcal G_{40}^{(8)}$
& $X_{15}$
& $(49433743,14045805)$
& $J_{15}$
& $(49495055,12839823)$ \\

$\widehat{\mathcal G}_{30}$
& $J_{15}$
& $(49495055,12839823)$
& $X_{15}^{X}$
& $(49430863,14088465)$ \\

$\widehat{\mathcal G}_{25}$
& $J_{15}$
& $(49495055,12839823)$
& $X_{15}$
& $(49433743,14045805)$ \\

$\mathcal G_{55}$
& $J_{30}^{++}$
& $(2455726444728097,841760069965664)$
& $J_{30}^{++}$
& $(2455726444728097,841760069965664)$ \\

\bottomrule
\end{tabularx}

\caption{Complete input parameters for the heterogeneous products in
Section~\ref{subsec:BPZheterogeneousImprovement}.
Part (a) records the ordered input gadgets, the codebook $J_0$,
and its decomposition parameters.
Part (b) records the two one-sided codebooks and their parameters.
All decompositions and neighborhood counts are relative to the
indicated left gadget.
The first row also specifies the fifteen-dimensional heterogeneous
construction used in the warmup.}
\label{tab:appendix-heterogeneous-inputs}
\end{table}

\medskip
\noindent\textbf{Reproducibility.}
The BPZ base data, combining rules, and terminal codes are anchored to commit
\texttt{aa21eeb12b75b0413d3fa9fb4208b5d0bf2c4d65} of the BPZ Lean
repository~\cite{BuysPolakZuiddamGitHub2026}. The relevant source files are
\texttt{ShannonBounds/}\allowbreak\texttt{BaseC7Data.lean},
\texttt{ShannonBounds/}\allowbreak\texttt{Substitutions.lean}, and
\texttt{ShannonBounds/}\allowbreak\texttt{TerminalCodes.lean}.
Machine-readable certificate files and exact-arithmetic verification
scripts are available in the repository
\cite{TandonC7Code2026}.
\end{appendices}

\end{document}